\documentclass[a4paper]{article}
\usepackage{amssymb}
\usepackage{amsmath}
\usepackage{amsthm}
\usepackage{amsfonts}
\usepackage{graphicx,epsfig,graphics}
\usepackage[left=2cm,top=3.5cm,right=2cm,nohead,foot=1cm]{geometry}
\usepackage{subfigure}
\usepackage{color}
\usepackage{hyperref}
\usepackage[normalem]{ulem}
\usepackage{authblk}

\usepackage{mathtools}

\newcommand{\ind}{1\hspace{-2.1mm}{1}}

\newcommand{\PP}{\mathbb{P}}
\newcommand{\Q}{\mathbb{Q}}
\newcommand{\D}{\mathrm{d}}
\newcommand{\Oo}{\mathcal{O}}
\newcommand{\Nn}{\mathcal{N}}
\newcommand{\E}{\mathrm{e}}
\newcommand{\BS}{\mathrm{BS}}

\newcommand{\pp}{\mathfrak{p}}
\newcommand{\qq}{\mathfrak{q}}
\newtheorem{theorem}{Theorem}

\newtheorem{corollary}[theorem]{Corollary}

\newtheorem{lemma}[theorem]{Lemma}

\newtheorem{proposition}[theorem]{Proposition}

\theoremstyle{definition}

\newtheorem{remark}[theorem]{Remark}
\newtheorem{example}[theorem]{Example}

\renewenvironment{proof}[1][Proof]{\noindent\textbf{#1.} }{\
\rule{0.5em}{0.5em}}
\numberwithin{equation}{section}
\numberwithin{theorem}{section}

\newcommand{\be}{\begin{equation}}
\newcommand{\ee}{\end{equation}}
\newcommand{\Prob}{\mathbb{P}}
\newcommand{\R}{\mathbb{R}}
\newcommand{\esp}{\mathbb{E}}
\newcommand{\eps}{\varepsilon}
\newcommand{\mino}{<}
\newcommand{\maj}{>}
\newcommand{\til}{~}

\graphicspath{ {./pics3/} }

\begin{document}

\title{Shapes of implied volatility with positive mass at zero\thanks{We thank Valdo Durrlemann, Archil Gulisashvili, Pierre Henry-Labord\`ere, Aleksandar Mijatovi\'c and Mike Tehranchi for stimulating discussions. 
We thank an anonymous referee for his/her important comments on the comparison of our results with that of Gulisashvili~\cite{ArchilAtom}.
SDM and CH are thankful to Nizar Touzi for his interest during the first phase of this work.
SDM and AJ acknowledge funding from the Imperial College Workshop Support Grant and the London Mathematical Society for the `Workshop on Large deviations and asymptotic methods in finance' (April 2013).
SDM and CH acknowledge funding from the research programs \emph{Chaire Risques Financiers}, \emph{Chaire March\'es en mutation} and \emph{Chaire Finance et d\'eveloppement durable}.
AJ acknowledges financial support from the EPSRC First Grant EP/M008436/1.
Corresponding author: demarco@cmap.polytechnique.fr
\newline \indent \emph{Key words and phrases}: Atomic distribution, heavy-tailed distribution, Implied Volatility, smile asymptotics, absorption at zero, CEV model.
\newline \indent \emph{2010 Mathematics Subject Classification}: AMS 91G20, 65C50.}}
\author[1]{S. De Marco}
\author[2]{C. Hillairet}
\author[2]{A. Jacquier}
\date{\today}
\affil[1]{CMAP, Ecole Polytechnique}
\affil[2]{ENSAE Paris Tech}
\affil[3]{Imperial College London}

\maketitle

\begin{abstract}
We study the shapes of the implied volatility when the underlying distribution has an atom at zero
and analyse the impact of a mass at zero on at-the-money implied volatility and the overall level of the smile.
We further show that the behaviour at small strikes is uniquely determined by the mass of the atom up to high asymptotic order, under mild assumptions on the remaining distribution on the positive real line.
We investigate the structural difference with the no-mass-at-zero case, 
showing how one can--theoretically--distinguish between mass at the origin and a heavy-left-tailed distribution.
We numerically test our model-free results in
stochastic models with absorption at the boundary, such as the CEV process, and in jump-to-default models.
Note that while Lee's moment formula~\cite{Lee} tells that implied variance is at most asymptotically linear in log-strike, other celebrated results for exact smile asymptotics such as~\cite{BenFr,GulForm} do not apply in this setting--essentially due to the breakdown of Put-Call duality.
\end{abstract}


\section{Introduction} 

Stochastic models are used extensively to price options and calibrate market data.
In practice, such data is often quoted, not in terms of option prices, but in terms of implied volatilities.
However, apart from the Black-Scholes model where the implied volatility is constant, 
no closed-form formula is available for most models.
Over the past decade or so, many authors have worked out approximations of this implied volatility, 
either in a model-free setting or for some specific models; 
these approximations are usually {only} valid in {restricted} regions, such as small and large maturities, or extreme strikes.
The latter have proved to be useful in order to extrapolate observed (and calibrated) data in an arbitrage-free way.
The celebrated moment formula by Lee~\cite{Lee} was a ground-breaking model-independent result in this direction; subsequent advances were made by Benaim and Friz~\cite{BenFr} and by Gulisashvili~\cite{GulForm}.
Denote $P(K)=\esp(K-S_T)^+$ the price of a Put option with strike~$K$ and maturity~$T$, 
where~$S$ is a positive random variable defined on some probability space with measure~$\Prob$.
Gulisashvili showed that the behaviour of the implied volatility $I(K)$ at small strikes is related to this Put price 
via the asymptotic formula~\cite[Corollary 5.12]{GulForm}
\be \label{e:generalTailWing}
I(K) = \sqrt{\frac{|\log K|}T} \sqrt{\psi \Bigl( \frac{\log P(K)}{\log K} - 1 \Bigr)}
+ \Oo \left( \Bigl(\log\frac{K}{P(K)}\Bigr)^{-1/2} \log \log \frac{K}{P(K)} \right),
\qquad \mbox{as } K \downarrow 0,
\ee
 where the continuous function $\psi: [0,\infty] \to [0,2]$ is defined by
\be \label{e:psiMomentFormula}
\psi(z)\equiv 2-4\left(\sqrt{z(z+1)}-z\right), 
\qquad \psi(\infty) = 0.
\ee
A similar formula, expressed in terms of the Call price $C(K)=\esp(S_T-K)^+$, holds as $K$ tends to infinity.
The expansion~\eqref{e:generalTailWing} is valid for every Put price function $P$ such that $P(K)>0$ for all $K > 0$,
which is equivalent to $\Prob(S_T < K) > 0$ for all $K > 0$.\footnote{If $\Prob(S_T < \overline{K})=0$ 
for some $0 < \overline{K} < S_0$, then $P(K) = 0$ for all $K \le \overline{K}$.
According to the definition of the implied volatility in~\cite{GulForm}, $I(K)$ is not defined for
such strikes; 
according to our extended definition~\eqref{e:implVolDef}, $I(K)$ is identically zero for all~$K \le \overline{K}$.} 
The formula~\eqref{e:generalTailWing} is obtained in two steps: first an asymptotic expansion for~$I(K)$ as~$K$ tends to infinity is given in terms of the Call price function; then the  expression as $K$ tends to zero is obtained via the Put-Call duality
$$
P(K) =
\esp\biggl[\frac{S_T}{S_0} \left(K \frac{S_0}{S_T}  - S_0 \right)^+\biggr]
= K \: \esp_{\Q} \biggl[\left(\frac{S_0}{S_T} - \frac{S_0}K \right)^+\biggr],
$$
where $\Q$  is a probability measure absolutely continuous with respect to $\Prob$, defined through its Radon-Nikodym density $\D\Q/\D\Prob \equiv S_T/S_0$.
The Put-Call symmetry above holds if (as implicitly assumed in~\cite{GulForm}) the law of the underlying asset price does not charge zero under~$\Prob$, i.e. if $\Prob(S_T=0)=0$.
The expansion~\eqref{e:generalTailWing}, then, is a priori not justified when $\Prob(S_T=0)>0$.

In certain stochastic models, the asset price is modelled with a stochastic process that accumulates mass at zero in finite time: 
this is for example the case for the Constant Elasticity of Variance (CEV) local volatility diffusion, 
whose fixed-time marginals have a continuous part and an atom at zero under certain parameters configurations (the same phenomenon appears for SABR, the stochastic volatility counterpart of CEV).
In the setting of default modelling, the class of structural models defines the default of a firm as the first time the firm's value hits a given threshold. In~\cite{Coc,Colgold}, the firm's value corresponds to its solvency ratio (logarithm of assets over debt), modelled via an Ornstein-Uhlenbeck process. 
An alternative approach, proposed by  Campi et al.~\cite{CampiSbuelz}, is to refer to the underlying equity process and define the default as the first time the process hits the origin: while the equity value remains deeply related to the firm's asset and debt balance sheet, such a modelling choice is easier to test than structural models, since equity data is more readily available. 
In the setting of~\cite{CampiSbuelz}, the equity process hits the origin either after a jump or in a diffusive way, the continuous-path part of the equity value being modelled by a CEV diffusion with a positive probability of absorption at zero.
Along the same line, we will consider in this paper asset prices that may either jump to zero, or hit zero along a continuous trajectory.

In this work, we study the impact of a mass at zero on at-the-money implied volatility and the overall level of the smile, and determine how the asymptotic behaviour of the implied volatility for small strikes is affected.
Concerning the second point, note that $\Prob(S_T=0)>0$ implies $q^*=0$, 
where $q^* \equiv \sup \{ q \ge 0 : \esp[S_T^{-q}] < \infty \}$ is the negative critical exponent of $S_T$.
Then, Lee's moment formula for small strikes yields, in full generality,
\begin{equation} \label{e:LeeLeft}
\limsup_{K \downarrow 0} \frac{\sqrt T I(K)}{\sqrt{|\log K|}} = \sqrt{\psi(q^*)} = \sqrt{2}.
\end{equation}
Tail-wing type refinements aim at finding conditions under which this $\limsup$ can be strengthened into a genuine limit, yielding the asymptotics $I(K) \sim \sqrt{2|\log K|/T}$ as~$K$ tends to zero: 
Benaim and Friz's result~\cite{BenFr} gives sufficient conditions, but is limited to the case $q^*>0$; Gulisashvili's result~\eqref{e:generalTailWing} applies to the case~$q^*=0$ and~$\Prob(S_T = 0)=0$, and allows to formulate necessary and sufficient conditions, as done in~\cite{GulIJTAF}.
\\
Denote $F(K) \equiv \Prob(S_T \le K)$, $\pp \equiv \PP(S_T=0)$ and $\qq \equiv \Nn^{-1}(\pp)$, 
where $\Nn^{-1}$ is the inverse Gaussian cumulative distribution.
The main results of this paper can be resumed as follows: if $\pp \maj 0$, then
\begin{itemize}
\item[(I)] the at-the-money implied volatility has a non-trivial lower bound: 
$I(S_0) \sqrt{T} \ge 2 \: \mathcal{N}^{-1} \Bigl(\frac12 (1+\pp) \Bigr)$. 
Moreover, for a large class of underlying distributions, the implied volatility smile is monotonically increasing with respect to $\pp$. 
The impact of the mass at zero is stronger for small strikes and asymptotically negligible for large strikes, in the sense that $I^\pp(K) \approx I^0(K) + \pp \Theta(K)$ 
for some function~$\Theta$ such that $\lim\limits_{K \downarrow 0} \Theta(K)=\infty$ 
and $\lim\limits_{K \uparrow \infty} \Theta(K)=0$. (See Theorem~\ref{t:impactOnSmile} for precise statements);
\item[(II)] the implied volatility satisfies $I(K) =
\sqrt{\frac{2 |\log K|}T} +\frac{\qq}{\sqrt T}
+ o(1)$ for small~$K$.
If $F(K) - F(0) = \Oo(|\log K|^{-1/2})$, 
the remainder term $o(1)$ is improved to $\Oo(|\log K|^{-1/2})$.
If moreover $F(K) - F(0) = \Oo(|\log K|^{-3/2})$, then the following asymptotic expansion holds:
\be \label{e:asymAtomIntro}
I(K) = \sqrt{\frac{2|\log K|}T} +\frac{\qq}{\sqrt{T}}
+ \frac{\qq^2 + 2}{2\sqrt{2T|\log K|}}
+ \frac{\qq}{4|\log K|\sqrt{T}} + \Oo\left(|\log K|^{-3/2}\right)
\qquad \text{as } K \downarrow 0.
\ee 
An estimate of the constant in front of the $\Oo\left(|\log K|^{-3/2}\right)$ error term is provided in 
Theorem~\ref{t:IVArchil}.
\end{itemize}

Slightly after the first version of this paper appeared, Gulisashvili~\cite{ArchilAtom} proved an asymptotic expansion for the left wing of the smile when the stock price has mass at the origin.
The main difference with~\eqref{e:asymAtomIntro} is that Gulisashvili's expansion~\cite{ArchilAtom} is written in terms of a non-explicit function of the strike $K$ (defined as the inverse of a given function--we 
refer the reader to Section~\ref{s:asymFormAtom} for precise definitions), 
while~\eqref{e:asymAtomIntro} only contain explicit functions of the strike and of the constant $\Nn^{-1}(\pp)$.
Formula~\eqref{e:asymAtomIntro} therefore allows to read-off the explicit dependence of the implied volatility in terms of the strike at a glance (at least up to a given asymptotic order), potentially allowing to improve parameterisations of the implied volatility smile in such a way to embed mass at zero 
(if one wishes to model default probability in this way, or otherwise to reproduce the left wing behaviour of stochastic models with absorption at zero, as we do in Section~\ref{s:cev} for the CEV model).
In particular,~\eqref{e:asymAtomIntro} highlights the presence in the expansion of a term proportional 
to $|\log K|^{-1}$, which was hidden in Gulisashvili's formulation~\cite{ArchilAtom}.

In order to measure the importance of the assumptions on the cumulative distribution function~$F(\cdot)$, 
let us note here that if the law of the stock price admits a density~$f$ in a right neighbourhood of zero, 
such that $f(K) = \Oo(K^{-a})$ for small~$K$, for some $a<1$, 
then $F(K)-F(0) = \Oo(K^{1-a})$. 
Therefore, the assumption $F(K) - F(0) = \Oo(|\log K|^{-3/2})$ is trivially fulfilled.

We organise the paper as follows: in Section~\ref{s:impactOnSmile}, 
we give the results related to item~(I) above. 
In Section~\ref{s:maximalSlope} we provide the first asymptotic estimates presented in~(II).
Building on the work of Gulisashvili~\cite{ArchilAtom}, we derive the explicit expansion~\eqref{e:asymAtomIntro} in Section~\ref{s:asymFormAtom}, and test this formula on several examples in Section~\ref{s:examples}.

\medskip

\textbf{Notations and preliminaries}.
\emph{Option prices}.
We fix here a maturity~$T\geq 0$, and shall therefore not indicate its dependence for simplicity.
We assume that $S_T$ is a non-negative integrable random variable on some probability space 
$(\Omega, \mathcal F, \Prob)$, with $\esp(S_T) = S_0>0$.
Risk-free interest rates are considered null, and option prices are given by expectations 
under the pricing measure~$\Prob$:
$C(K) \equiv \esp[(S_T-K)^+]$
and
$P(K) \equiv \esp[(K - S_T)^+]$ denote the prices of European Call and Put options with strike~$K\geq 0$ and maturity~$T$.
$C_{\BS}(K; S_0, \sigma)$ and $P_{\BS}(K; S_0, \sigma)$ denote the corresponding Call and Put prices in the Black-Scholes model with volatility parameter~$\sigma$:
\begin{align}
C_{\BS}(K; S_0, \sigma) & \equiv 
\left \{ \begin{array}{ll}
S_0 \Nn\left(d_1(\log(K/S_0)), \sigma\right) - K \Nn\left(d_2(\log(K/S_0)), \sigma\right), & \mbox{if } \sigma > 0
\\
(S_0 - K)^+,  & \mbox{if } \sigma = 0
\end{array}
\right.\label{defCall}
\\
P_{\BS}(K; S_0, \sigma) & \equiv 
\left \{ \begin{array}{ll}
K \Nn\left(-d_2(\log(K/S_0)), \sigma\right) - S_0 \Nn\left(-d_1(\log(K/S_0)), \sigma\right), & \mbox{if } \sigma > 0
\\
(K-S_0)^+,  & \mbox{if } \sigma = 0
\end{array} \label{defPut}
\right.
\end{align}
where
$\label{eq:d12} d_{1,2}(x,\sigma)\equiv \frac{-x}{\sigma \sqrt{T}} \pm \frac{1}{2}\sigma \sqrt{T}$,
and $\Nn$ is the standard Gaussian cumulative distribution function
$\Nn(d) \equiv \int_{-\infty}^d n(z)\D z$, 
with $n(z) \equiv (2\pi)^{-1/2}\E^{-\frac{1}{2}z^2}$.
When the spot price $S_0$ is fixed, it should not generate any confusion to use the same notation~$C_{\BS}$ 
and~$P_{\BS}$ for the (normalised) option prices with log-moneyness $x = \log(K/S_0)$:
$$
S_0 C_{\BS}(x,\sigma) \equiv C_{\BS}(K_x; S_0, \sigma)
\qquad\text{and}\qquad
S_0 P_{\BS}(x,\sigma) \equiv P_{\BS}(K_x; S_0, \sigma),
$$
where
$K_x \equiv S_0 \E^x$.

\emph{Implied volatility}.
The implied volatility\footnote{The implied volatility obviously depends on~$T$ but, 
since the latter is fixed, we shall also drop this dependence in the notation.} 
$I(x)$ is defined as the unique solution in $[0,\infty)$ to the equation
\be \label{e:implVolDef}
S_0 C_{\BS}(x,I(x)) = C(K_x)
\ee
Note that $I(x)$ is a strictly positive real number when~$C(K_x)$ satisfies the strict arbitrage
bounds $(S_0 - K_x)^+ < C(K_x) < S_0$, and it is zero if $C(K_x)= (S_0 - K_x)^+$.
With a slight abuse of notation, and where explicitly stated, we might also denote $I(K)=I(\log(K/S_0))$ the implied volatility as a function of strike.

\emph{Function asymptotics.}
For a function $g$ defined on a punctured neighbourhood of $x_0 \in [-\infty,\infty]$, we write
\begin{itemize}
\item $f=o(g)$ (resp. $f=\Oo(g)$) when $f(x)=g(x) \phi(x)$ in a neighbourhood of~$x_0$, 
for some function~$\phi$ such that $\lim_{x\to x_0}\phi(x) = 0$
(resp. for some $\phi$ bounded around $x_0$);
\item $f(x) = g(x) + \Oo(h(x))$ as $x$ approaches ~$x_0$ when $f-g = \Oo(h)$ in a neighbourhood of~$x_0$;
\item $f(x)\sim g(x)$ around~$x_0$ if~$g$ is non vanishing and $\lim_{x \to x_0}f(x)/g(x)=1$.
\end{itemize}
\emph{Mass at zero and first-order behaviour of the Put price}.
By Fubini's theorem, $P(K) = \int_0^{K} F(y) \D y$, and hence
\be \label{e:putAsympt}
\lim\limits_{K \downarrow 0} \frac{P(K)}K
=
\lim\limits_{K \downarrow 0} \frac1K \int_0^{K} F(y) \D y
= F(0) = \Prob(S_T = 0).
\ee
Finally, we shall denote $\pp = \PP(S_T = 0)$ the mass at zero, and $\qq = \Nn^{-1}(\pp)$.


\section{Overall impact on the smile} \label{s:impactOnSmile}

In this section, we investigate the impact that a mass at zero has on the overall behaviour of the implied volatility.
Intuitively, the impact should be more important on the left part of the smile and less significant on the right part.
Concerning the at-the-money behaviour, we will now show that a mass at zero imposes a lower bound on the level of the implied volatility.
\begin{proposition}
The following lower bound holds for the at-the-money implied volatility:
\be \label{e:ATMlowBound}
I(0) \sqrt{T} \ge 2 \: \mathcal{N}^{-1} \Bigl(\frac12 (1+\pp) \Bigr).
\ee
\end{proposition}
\begin{proof}
When $K=S_0$, the Black-Scholes formula~\eqref{defPut} degenerates to $P_{\BS}(S_0; S_0, \sigma) = S_0\mathcal{N}(\frac{\sigma \sqrt T}2) - S_0 \mathcal{N}(-\frac{\sigma \sqrt T}2) = S_0 \left(2\mathcal{N}(\frac{\sigma \sqrt T}2) -1\right)$, and therefore
\[
\frac{I(0) \sqrt{T}}2 = \mathcal{N}^{-1}\left(\frac12 \left(1+\frac{P(S_0)}{S_0}\right) \right).
\]
Since $P(S_0) = \esp[(S_0-S_T)^+] \ge S_0 \pp$, the proposition follows.
\end{proof}
\medskip

\begin{remark}
When $\pp=0$,~\eqref{e:ATMlowBound} corresponds to the trivial bound $I(0) \ge 0$.
When $\pp \maj 0$, the lower bound in~\eqref{e:ATMlowBound} explains why the implied volatility smiles generated by a distribution with a mass at zero are typically very high, 
as can be seen in all the smiles plotted in Figures~\ref{f:smilesAffine}(d),~\ref{f:smilesBSAtom},~\ref{f:CEVPlotsSmallMass}(a) and~\ref{f:CEVPlotsBigMass}(a).
In order to have a better idea of the magnitude of the lower bound, we can use an approximation of the inverse Gaussian cdf~\cite{AAnormCdf} $\mathcal{N}^{-1}(u) \approx \sqrt{ -\log(1 - (2u-1)^2) / \sqrt{\pi/8}}$, 
which yields $\mathcal{N}^{-1} \Bigl(\frac12 (1+\pp) \Bigr) \approx \sqrt{ -\log(1 - \pp^2) / \sqrt{\pi/8}}$.
A Taylor expansion of the log function around $\pp=0$ gives $2 \mathcal{N}^{-1} \Bigl(\frac12 (1+\pp) \Bigr) \approx \frac1{(\pi/8)^{1/4}} \pp \sqrt{1 + \frac{\pp^2}2} \approx 2.5 \, \pp \sqrt{1 + \frac{\pp^2}2}$ (meaning that, when $\pp=0.2$, the lower bound is around $50\%$).
See Figure~\ref{f:ATMVolLowerBound} for a numerical example.
\end{remark}

\begin{figure}[t]
     \begin{center}
            \includegraphics[scale=0.5]{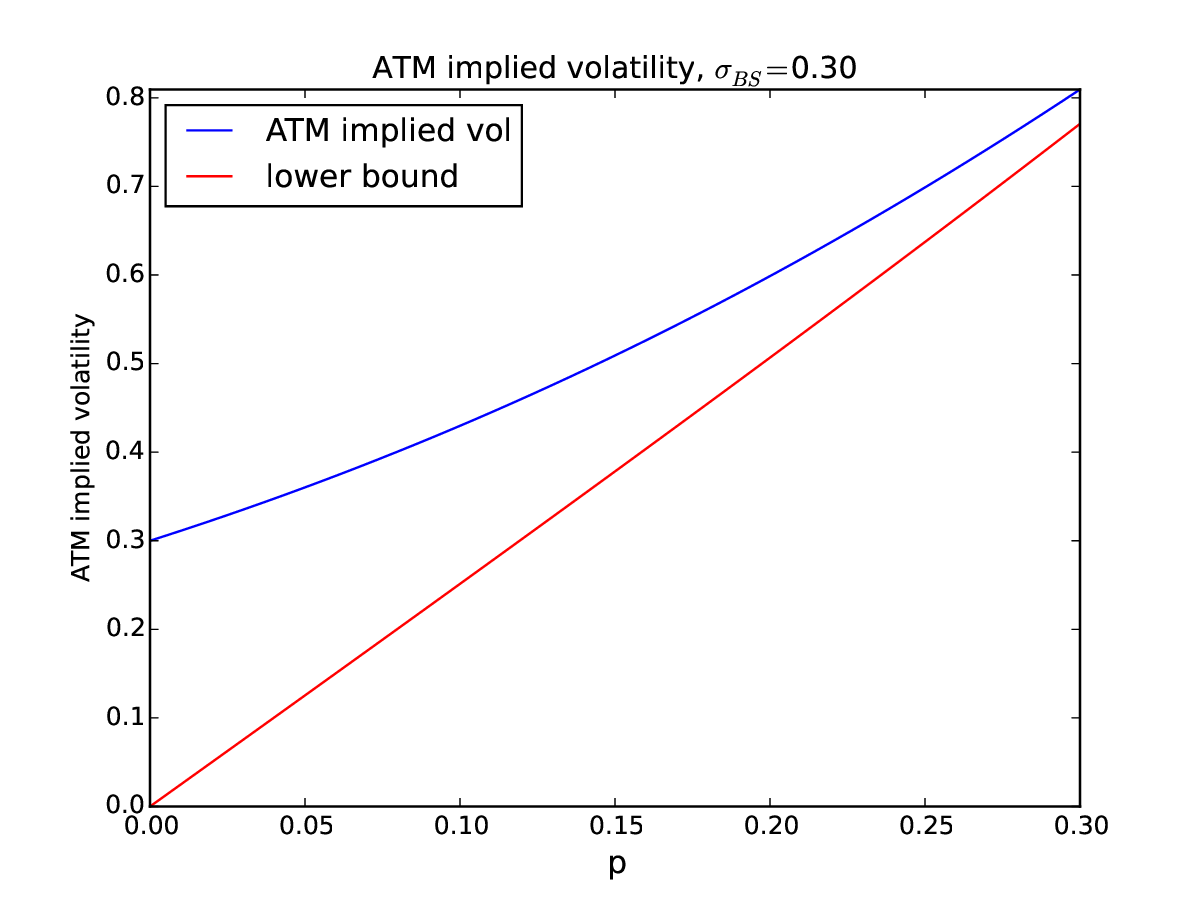}
    \end{center}
\caption{At-the-money implied volatility for the Merton model with jump-to-default (see Section~\ref{s:merton0}). The underlying distribution is $\mu = \pp \delta_0 + (1-\pp) \tilde \mu_{\BS}$, 
where $\tilde \mu_{\BS}$ is a Black-Scholes distribution with volatility parameter $\sigma=0.3$ 
and spot value $1/(1-\pp)$. 
The red line shows the lower bound~\eqref{e:ATMlowBound}.
}
\label{f:ATMVolLowerBound}
\end{figure}

We are now going to show that the introduction of a mass at zero in an asset price model has the effect of lifting the whole smile (that is: simultaneously for all strikes).
We will focus on a certain (large) class of distributions.
Precisely, we consider
the family of random variables $\tilde S = s X$, indexed by their mean value $s = \esp[\tilde S]$, where $X$ is a positive random variable such that $\esp[X]=1$ and $\Prob(X=0)=0$.
This setting covers the case of stochastic volatility models ($d\tilde S_t = \tilde S_t \sigma_t dW_t$, $\tilde S_0 = s$) and exponential L\'evy models.
We denote by~$\tilde{\mu}_s$ the distribution of~$\tilde S$ on $(0,\infty)$. 
This framework provides the reference model, which is then enhanced with a mass at zero by setting
\be \label{e:enhancedModel}
\mu^{(\pp)} = \pp\delta_{0} + (1-\pp) \tilde{\mu}_{\frac{S_0}{1-\pp}}
\ee
for the distribution $\mu^{(\pp)}$ of~$S_T$,
where~$\delta_0$ denotes the Dirac distribution at the origin.
Equation~\eqref{e:enhancedModel} covers the class of models with an independent jump to default, 
presented in Section~\ref{s:j2d}. 
Note that the mean value of~$\tilde{\mu}$ in~\eqref{e:enhancedModel} is imposed by the condition 
$\esp[S_T] = S_0$.
We denote $\tilde F_{s}(K) = \int_0^K \tilde \mu_s(\D y)$ the cdf of $\tilde \mu_s$, 
$P^{(\pp)}(K) = \int (K - y)^+ \mu^{(\pp)}(\D y)$ the price of the Put option with strike~$K$, 
and~$I^{(\pp)}(\cdot)$ the corresponding implied volatility, defined by $S_0 P_{\BS}(x,I^{(\pp)}(x)) = P^{(\pp)}(K_x)$.

\begin{theorem} \label{t:impactOnSmile}
Assume the asset price distribution is given by~\eqref{e:enhancedModel}. Then
\begin{itemize}
\item[(i)] the function $\pp \mapsto I^{(\pp)}(x)$ is increasing on $[0,1)$ for all $x \in \R$;
\item[(ii)] the function $\Delta I^{(\pp)}:x\mapsto I^{(\pp)}(x) - I^{(0)}(x)$ satisfies
\be \label{e:deltaI}
\sqrt T \Delta I^{(\pp)}(x) \sim \pp \: \theta(x)
\qquad \text{for all $x\in \R$, as } \pp \text{ tends to }0,
\ee
where $\theta(x) = \displaystyle \frac{1 - \tilde F_{S_0}(K_x^-)}{n(d_2(x,I^{(0)}(x))}$
satisfies
$\lim\limits_{x \downarrow -\infty} \theta(x) = +\infty$.
\end{itemize}
\end{theorem}
\begin{proof}
\emph{(i)} 
It is clear from the definition of the Put price that 
$P^{(\pp)}(K) = \pp K + (1-\pp) \int_0^K \tilde F_{\frac{S_0}{1-\pp}}(y) dy$.
Using the definition of $\tilde S$, we immediately have 
$\tilde F_{s'}(y) = \Prob(s' X \le y) = \tilde F_s \left(\frac s{s'} y \right)$ for all $s, s' \maj 0$.
Therefore, $P^{(\pp)}(K) = \pp K + (1-\pp) \int_0^K \tilde F_{S_0}((1-\pp)y) dy = \pp K + \int_0^{(1-\pp)K} \tilde F_{S_0}(y) dy$. 
Then, for any $0\leq \pp_1 \leq \pp_2 <1$, 
$$
P^{(\pp_2)}(K) - P^{(\pp_1)}(K)
= (\pp_2 - \pp_1)(K)
- \int_{(1-\pp_2)K}^{(1-\pp_1)K} \tilde F_{S_0}(y) dy
=
\int_{(1-\pp_2)K}^{(1-\pp_1)K} \left(1 - \tilde F_{S_0}(y) \right) dy,
$$
so that, for every $K \maj 0$ the map $\pp \mapsto P^{(\pp)}(K)$ is increasing, 
and so is $\pp \mapsto I^{(\pp)}(x)$ on~$[0,1)$, for every $x\in\R$.

\emph{(ii)}
Let us denote $K = K_x$ for simplicity.
By the definition of $\Delta I^{(\pp)}$, we have $S_0 P_{\BS}(x,I^{(0)}(x)+\Delta I^{(\pp)}(x)) = P^{(\pp)}(K)$.
It is clear, for every $x$, $\Delta I^{(\pp)}(x) \to 0$ as $\pp \to 0$ ($I^{(\pp)}$ is a continuous function of $\pp$), therefore 
\be \label{e:deltaI1} 
P_{\BS}(x,I^{(0)}(x)) + \partial_{\sigma}P_{\BS}(x,I^{(0)}(x)) \Delta I^{(\pp)}(x) \left(1+o(1)\right) = \frac1{S_0} P^{(\pp)}(K)
\qquad \mbox{as } \pp \to 0.
\ee
In~(i) of the present proof we have shown that 
$P^{(\pp)}(K) = \pp K + \int_0^{(1-\pp)K} \tilde F_{S_0}(y) dy$. 
Equation~\eqref{e:deltaI1} then yields
\[
\partial_{\sigma}P_{\BS}(x,I^{(0)}(x)) \Delta I^{(\pp)}(x) \left(1+o(1)\right) = 
\frac1{S_0} \left( P^{(\pp)}(K) - P^{(0)}(K) \right)
=
\frac{\pp K}{S_0} \biggl(1  - \frac 1{\pp K} \int_{(1-\pp)K}^K \tilde F_{S_0}(y) dy \biggr)
\]
or yet
\[
\Delta I^{(\pp)}(x) \sim 
\frac{\pp K}{S_0 \partial_{\sigma}P_{\BS}(x,I^{(0)}(x))} (1 - \tilde F_{S_0}(K^-))
\qquad \mbox{as } \pp \text{ tends to zero}.
\]
Using the well-known expression $S_0 \partial_{\sigma}P_{\BS}(x,\sigma) = \sqrt T K n(d_2(x,\sigma))$, 
we obtain~\eqref{e:deltaI}.
Finally, since the limit
$\lim_{x \downarrow -\infty} d_2(x,I(x)) = +\infty$ holds for any distribution without mass at zero
(Lemma~\ref{l:d2Estimate}), we also have $\lim_{x \downarrow -\infty} n(d_2(x,I^{(0)}(x))) = 0$.
Since $\lim\limits_{x \downarrow -\infty} (1 - \tilde F_{S_0}(S_0 e^x))=1$, 
then $\lim\limits_{x \downarrow -\infty} \theta(x) = +\infty$.
\end{proof}
\medskip

\noindent
Point (ii) in Theorem~\ref{t:impactOnSmile} shows that the impact of a mass at zero is stronger for small strikes.

\begin{remark}
Assume that the reference model $\tilde \mu$ follows the Black-Scholes distribution with volatility parameter $\sigma$, so that $\tilde F_{S_0}(K_x) = \mathcal{N}(-d_2(x,\sigma))$.
Then, as $x \to \infty$,
\[
\begin{aligned}
\theta(x) = \frac{\mathcal{N}(d_2(x,\sigma))}{n(d_2(x,\sigma))}
\sim \frac 1{|d_2(x,\sigma)|}
\sim \frac {\sigma \sqrt T}{x},
\end{aligned}
\] 
where we used the well-known expansion $\mathcal{N}(z) = n(z) \frac 1{|z|} (1+o(1))$ as $z \to -\infty$. 
On this example, we see that the impact of a mass at zero (quantified by the function $\theta(\cdot)$) becomes asymptotically negligible as~$K$ tends to infinity.
This phenomenon can be seen clearly in our numerical tests on the Merton model with jump-to-default in Section~\ref{s:merton0}, see Figure~\ref{f:smilesBSAtom}.
\end{remark}

\section{Asymptotic estimates}
\label{s:maximalSlope}

The moment formula~\eqref{e:LeeLeft} guarantees that $\limsup_{x\downarrow -\infty} I(x)^2 T/|x|$ is strictly smaller than $2$ when $q^*>0$.
The two situations where the $\limsup$ reaches the level $2$, then, are $q^*=0$ and $\pp=0$
(heavy left tail but no mass at zero), and $\pp>0$.
The former case is considered by Gulisashvili~\cite{GulIJTAF}; our focus is on the latter.

\begin{example} \label{ex:HullWhite}
In the Hull-White stochastic volatility model, the stock price process satisfies the stochastic differential equation $\D S_t=S_t |Z_t| \D W_t$, with $S_0 >0$, and $Z$ is a lognormal process satisfying $\D Z_t = \nu Z_t \D t + \xi Z_t \D B_t$, with {$W$ and $B$ two correlated Brownian motions} $\D \langle W,B\rangle_t = \rho \D t$.
For all $T \ge 0$, $S_T$ is a strictly positive and integrable random variable. As shown in~\cite{GulStHW}, all the moments of $S_T$ of order smaller than zero or larger than one are infinite: $q^* = 0 = p^* := \sup\{p\geq 0: \mathbb{E}(S_T^{1+p})<\infty\}$.
\end{example}

\subsection{First-order behaviour} \label{s:firstOrderAtom}

In the spirit of the tail-wing formula~\cite{BenFr}, the expansion~\eqref{e:generalTailWing} 
allows to convert Lee's moment formula into an asymptotic equivalence.
This requires to study the behaviour of $\psi \bigl(\frac{\log P(K)}{\log K} - 1\bigr)$ for small~$K$.
The identity
\be \label{e:limInf}
\liminf_{K\downarrow 0} \frac{\log P(K)}{\log K} = 1+q^*
\ee
is given in~\cite[Lemma 4.5]{GulForm}.
In general, $\limsup_{K \downarrow 0} \frac{\log P(K)}{\log K}$ is not necessarily equal to $1+q^*$. 
In Gulisashvili~\cite{GulIJTAF}, conditions on the Put price function equivalent to $\limsup_{K \downarrow 0} \frac{\log P(K)}{\log K}=1+q^*$ are given.
Let us recall Gulisashvili's result for the case $q^*=0$ of interest to us:

\begin{theorem}[Theorem 3.6 in~\cite{GulIJTAF}] \label{t:thmGulIJTAF}
If $q^*=0$ and $\pp=0$, then the following statements are equivalent:
\begin{itemize}
\item[(i)]
$\limsup_{K \downarrow 0} [\log P(K)/\log K] = 1$;
\item[(ii)]
$\sqrt T I(x) \sim \sqrt{2|x|}$, as $x$ tends to $-\infty$;
\item[(iii)] there exist $\underline{K}>0$ and a regular varying\footnote{
A function $f$ is regularly varying of order $\alpha \in \R$ if it is defined on some neighbourhood of infinity, measurable, and such that the ratio $\frac{f(\lambda x)}{f(x)}$ converges to $\lambda^{\alpha}$ as $x$ tends to infinity,
for every $\lambda>0$.} function $h$ of order $-1$ such that $h(\frac{1}{K}) \le P(K)$ for $K\in (0,\underline{K})$.
\end{itemize}
\end{theorem}

In light of~\eqref{e:generalTailWing}, Condition~(ii) in the previous theorem is equivalent to 
$\lim_{K \downarrow 0} \psi\bigl(\frac{\log P(K)}{\log K} - 1\bigr)=2$, 
or equivalently $\lim_{K \downarrow 0} \frac{\log P(K)}{\log K}=1$ by continuity of $\psi^{-1}$: considering~\eqref{e:limInf}, (ii) is then equivalent to (i).
For the proof of the equivalence between (ii) and (iii) we refer to~\cite{GulIJTAF}: the approach is first to show the equivalence between $\sqrt T I(x) \sim \sqrt{2x}$ for large~$x$ and a condition on the Call price function analogous to~(iii) (see~\cite[Theorem 3.2]{GulIJTAF}), and second to apply the Put-Call symmetry in order to transfer the result from the right to the left wing.
Because of the lack of Put-Call symmetry when the law of the stock price has a mass at zero, 
this approach is a priori not justified when $q^*=0$ and $\pp>0$, 
just as it happens for the asymptotic formula~\eqref{e:generalTailWing}.
We shall get back to this point in Remark~\ref{r:thmGulIJTAF}.
Let us state a preliminary result on the behaviour of the implied volatility when $\pp>0$.
Similarly to Theorem~\ref{t:thmGulIJTAF}, the following proposition reinforces~\eqref{e:LeeLeft} to a true limit.

\begin{proposition} \label{p:asymptAtom}
If $\pp>0$, then
$\sqrt T I(x) \sim \sqrt{2 |x|}$ as $x$ tends to $-\infty$.
\end{proposition}

Proposition~\ref{p:asymptAtom} follows from a stronger statement given in Theorem~\ref{t:secondOrder} below.

\begin{remark} \label{r:thmGulIJTAF}
When $\pp>0$, the function $h_1(K) \equiv \pp/K$ is regularly varying of index~$-1$.
Since $P(K) = \esp[(K-S_T)^+] \ge \pp K$ for every $K\ge0$, 
Theorem~\ref{t:thmGulIJTAF}(iii) is satisfied with the function~$h_1$.
Moreover, note that~\eqref{e:putAsympt} implies $\log(P(K))=\log(K) + \Oo(1)$ as $K \downarrow 0$, 
or equivalently $\lim_{K\downarrow 0}\log(P(K))/\log(K) = 1=1+q^*$.
Then, in view of Proposition~\ref{p:asymptAtom}, Theorem~\ref{t:thmGulIJTAF} turns out to be true also in the case $\pp>0$.
\end{remark}

\begin{remark}
A positive mass at zero implies infinite  expectation for certain payoffs, such as $\log \frac{S_T}{S_0}$.
Indeed, the right-hand side of the model-free replication formula for the log contract,
\be \label{e:formulaLogContract}
-\esp \left(\log \frac{S_T}{S_0} \right)
=
\int_{0}^{S_0}\frac{P(K)}{K^2}\D K + \int_{S_0}^{\infty}\frac{C(K)}{K^2}\D K,
\ee
is infinite, in light of~\eqref{e:putAsympt}.
This warns about the use of~\eqref{e:formulaLogContract}--typically applied to quote the fair strike of a continuously monitored variance swap under a stochastic volatility assumption as in~\cite{Derman}--in 
models where $\pp \maj 0$.
\end{remark}

\subsection{Detecting the mass of the atom: the second-order behaviour}

In the previous section, we saw that the dimensionless implied volatility $\sqrt T I(x)$ is asymptotic to $\sqrt{2|x|}$ as $x \downarrow -\infty$ if and only if $q^*=0$ and one of the equivalent conditions (i) or (iii) in Theorem~\ref{t:thmGulIJTAF} is fulfilled (which is always true when $\pp>0$ by Remark~\ref{r:thmGulIJTAF}).
The next step is to understand how the difference $\sqrt T I(x)-\sqrt{2|x|}$ behaves.
The behaviour of the right wing ($x \uparrow +\infty$) has been studied by Lee~\cite[Lemma 3.1]{Lee}, 
who showed that $\sqrt T I(x)-\sqrt{2|x|}$ is negative for $x$ large enough, and subsequently refined by Rogers and Tehranchi~\cite[Theorem 5.3]{RogTeh}, who proved that
$\lim_{x \uparrow +\infty} (\sqrt T I(x) - \sqrt{2x}) = -\infty$
for every $T>0$.
For the left wing, the situation is different, and the qualitative behaviour of the second-order term depends on the presence of a mass at zero.

\begin{theorem} \label{t:secondOrder}
If $\pp=0$, then 
\be \label{e:secondOrderNoAtom}
\lim_{x \downarrow -\infty} \left( \sqrt T I(x) - \sqrt{2|x|} \right) = -\infty.
\ee
On the contrary, if $\pp>0$, then 
\be \label{e:firstExpansion}
\sqrt{T} I (x) =
\sqrt{2 |x|}
+ \qq
+ \frac{\sqrt{2 \pi}\exp\left(\frac{1}{2}\qq^2\right)}{K_x} \int_0^{K_x} (F(y)-F(0))\D y
+\frac{\overline \chi(x)}{2\sqrt{2|x|}}
+\Oo\biggl(\frac{F(K_x)-F(0)}{\sqrt{|x|}}\biggr),
\ee
where $\overline \chi$ is a function satisfying $\qq^2 \le \overline \chi(x)$ for all $x < 0$ and 
$\limsup\limits_{x \downarrow -\infty} \overline \chi(x) \le \qq^2 + 2\E^{\frac{1}{2}\qq^2}$.
\end{theorem}

Before giving the proof of Theorem~\ref{t:secondOrder}, let us state an immediate corollary.
Consider the following assumptions on the cumulative distribution function of $S_T$, as $K$ tends to zero:
\be \label{e:assumptionsCdf}
\begin{aligned}
(i) \quad F(K) - F(0) = \Oo\left(|\log(K/S_0)|^{-1/2}\right);
\\
(ii) \quad F(K) - F(0) = o\left(|\log(K/S_0)|^{-1/2}\right).
\end{aligned}
\ee

\begin{corollary} \label{c:secondOrder}
Assume $\pp>0$.
Then \footnote{The limit~\eqref{e:secondOrderNoAtom} also appeared in the preprint by Fukasawa~\cite{Fukasawa}, but was not reported in the published version of that paper, 
and the limit~\eqref{e:secondOrderAtom} appeared in the conference presentation 
of Tehranchi~\cite{TehPresentation}. 
We thank both authors for pointing this out to us, and Mike Tehranchi for sharing with us his proof of~\eqref{e:secondOrderAtom}.}
\be \label{e:secondOrderAtom}
\lim_{x \downarrow -\infty} \Bigl(\sqrt{T} I(x) - \sqrt{2|x|}\Bigr) = \qq;
\ee
if moreover~\eqref{e:assumptionsCdf}$(i)$  holds, then
\be \label{e:secondOrderExplicitError}
\sqrt{T} I (x) =
\sqrt{2 |x|}
+\qq
+\Oo\left(|x|^{-1/2}\right);
\ee
if moreover~\eqref{e:assumptionsCdf}$(ii)$ holds, then
\be \label{e:secondOrderWithErrorEstimate} 
\sqrt{T} I (x) =
\sqrt{2 |x|}
+\qq
+\frac{\chi(x)}{2\sqrt{2|x|}},
\ee
where $\chi$ is a function satisfying $\qq^2 \le \liminf\limits_{x \downarrow -\infty} \chi(x) 
\le \limsup\limits_{x \downarrow -\infty} \chi(x) \le \qq^2 + 2\E^{\frac{1}{2}\qq^2}$.
\end{corollary}
\begin{proof}
The limit~\eqref{e:secondOrderAtom} follows from~\eqref{e:firstExpansion}.
The other statements are immediate, noticing that $\overline{\chi}(x)=\Oo(|x|^{-1/2})$ and $F(K_x)-F(0)=\Oo(|x|^{-1/2})$ (resp. $o(|x|^{-1/2})$) under condition $(i)$ (resp. $(ii)$).
\end{proof}
\medskip

The following comments emphasise the relevance of Theorem~\ref{t:secondOrder}.

\begin{itemize}
\item In light of~\eqref{e:secondOrderNoAtom} and~\eqref{e:secondOrderAtom}, when the left asymptotic slope of the smile is maximal 
($\lim_{x \downarrow -\infty} T I(x)^2/|x| =2$), the difference between an underlying distribution that has a mass at the origin and one that does not can be seen at the second order in implied volatilities at small strikes.

\item The inspection of~\eqref{e:secondOrderExplicitError} reveals a `phase transition' in the shape of the implied volatility at the second-order: when  $\pp\neq 1/2$, the implied volatility has the form 
$\sqrt T I(x) =\sqrt{2 |x|}+const.+\Oo(|x|^{-1/2})$; when $\pp=1/2$, the constant vanishes and the expansion reduces to
$\sqrt T I(x) = \sqrt{2 |x|} +\Oo(|x|^{-1/2})$ under Condition~\eqref{e:assumptionsCdf}$(i)$.
In the latter case, the rate of convergence of the `normalised' implied volatility $I (x)\sqrt T/\sqrt{|x|}$ to its limit~$\sqrt2$ is $|x|^{-1}$ instead of $|x|^{-1/2}$.

\item The role played by the cumulative distribution function $F$ in~\eqref{e:firstExpansion} highlights a radical difference with the no-mass-at-zero case.
In the classical left tail-wing formula~\cite{BenFr},
\be \label{e:ltw}
\sqrt T I(x) \sim \sqrt{|x|} \sqrt{ \psi\left(\frac{-\log F(K_x)}{|x|}\right)},
\qquad \mbox{as } x \downarrow -\infty,
\ee
where $\psi$ is defined in~\eqref{e:psiMomentFormula}.
Note that the logarithm of the cdf $F(K_x)$ appears in the formula, instead of the cdf itself as in~\eqref{e:firstExpansion}.
In many stochastic volatility models, such as Heston and Stein-Stein, 
the cdf of the stock price satisfies~\cite{GulSt, Refined},
\be \label{e:cdfAsympt}
F(K_x) = A \: \E^{-\alpha_1 |x| + \alpha_2 \sqrt{|x|}} \: |x|^{\gamma} (1+\Oo(|x|^{-1/2}))
\qquad \mbox{as } x \downarrow -\infty,
\ee
for some constants $A,\alpha_1,\alpha_2>0$ and $\gamma \in \R$.
Therefore, $-\log F(K_x)/|x| = \alpha_1 + \Oo(|x|^{-1/2})$,
and~\eqref{e:ltw} returns--as expected--the leading-order square root behaviour $I(x) \sim \sqrt{\psi(\alpha_1)|x|}$
(subsequent refinements are of course possible using the precise asymptotics~\eqref{e:cdfAsympt}, as done in~\cite{Refined,GaoLee}; see also Remark~\ref{r:stochVol} below for further discussions).
For any distribution such that $F(K_x) - F(0)$ behaves as the right-hand side of~\eqref{e:cdfAsympt}, the terms of order equal or lower than $F(K_x)-F(0)=\Oo(\E^{-\alpha_1|x|})$ in~\eqref{e:firstExpansion} go to zero much faster than the $|x|^{-1/2}$ term.
\end{itemize}

\begin{remark}[On the limit~\eqref{e:secondOrderAtom}] \label{r:onSecondOrder}
Lemma 3.3 in~\cite{Lee} asserts that there exists $x^*$ such that $\sqrt T I(x) - \sqrt{2|x|}<0$ for all $x<x^*$ if and only if $0 \le \pp<1/2$.
In light of the estimate~\eqref{e:secondOrderAtom}, the difference $\sqrt T I(x) - \sqrt{2|x|}$ converges to a negative constant when $0<\pp<1/2$, to a positive constant when $\pp>1/2$, and to zero when $\pp=1/2$.
In diffusion models with absorption at zero such as the CEV model in Section~\ref{s:cev}, 
the case $\pp \approx 0$ corresponds to small $T$, while large values of $\pp$ (close to~$1$) 
correspond to large~$T$.
\end{remark}

The proof of Theorem~\ref{t:secondOrder} is based on the following two lemmas.

\begin{lemma} \label{l:remainderEstimate}
Let $R(K) \equiv K^{-1}P(K) - \pp$, for
$K>0$.
Then 
$R(K) = \frac{1}{K} \int_0^K (F(y)-F(0))\D y$.
In particular, $0 \le R(K)\le F(K)-F(0)$ for all $K>0$.
\end{lemma}
\begin{proof}
We have $R(K) = \frac1{K}P(K) - \pp = \frac1{K}\int_0^{K} F(y) \D y - F(0) = \frac1{K}\int_0^{K} (F(y)-F(0))\D y$.
The final estimate on~$R$ follows from the monotonicity of~$F$. 
\end{proof}

\begin{lemma}\label{l:d2Estimate}
If $\pp=0$, then $d_2(x, I(x)) \to +\infty$ as $x \downarrow -\infty$.
If $\pp>0$, then there exists a function $\varphi: (-\infty,0) \to (0,\infty)$ such that 
$\limsup_{x \downarrow -\infty} \varphi(x) \sqrt{2|x|} \le 1$, and such that the following estimate holds as $x \downarrow -\infty$:
$$
d_2(x, I(x)) = - \qq -\E^{\frac{1}{2}\qq^2} \varphi(x) - \frac{\sqrt{2\pi}}{K_x} \E^{\frac{1}{2}\qq^2}
\int_0^{K_x} (F(y)-F(0)) \D y
+ \Oo \left(\frac{1}{|x|}\right)
+ \Oo\left(\left(F(K_x)-F(0)\right)^2\right).
$$
\end{lemma}

\begin{remark}
The limit $\lim_{x\downarrow-\infty}d_2(x, I(x)) = - \qq $
was given in~\cite[Theorem 1]{OOUY} under the additional assumption 
that~$I$ is differentiable and that the derivative has a limit at $-\infty$ (see their Assumption~1). Later on, Fukasawa~\cite{Fukasawa} proved that the limit holds without these assumptions.
\end{remark}

\begin{proof}[Proof of Lemma~\ref{l:d2Estimate}]
The identity 
$C_{\BS}(-x,I(x))=\esp\Bigl[\Bigl(1-\E^{-x}\frac{S_T}{S_0}\Bigr)^+\Bigr]$
follows from Put-Call parity.
Then, for every $x<0$,
\be \label{e:d1parityEstimate}
\begin{aligned}
\Nn(d_{1}(-x,I(x)))
&= C_{\BS}(-x,I(x)) + \E^{-x} \Nn(d_{2}(-x,I(x)))
\\
&= \mathbb{P}(S_T=0) +\mathbb{E}\left[\left(1-\E^{-x}\frac{S_T}{S_0}\right)^+\ind_{\{S_T>0\}}\right]
+\E^{-x} \Nn(d_{2}(-x,I(x)))
\\
&= \pp + R(K_x) + \bar\varphi(x),
\end{aligned}
\ee
where~$R(\cdot)$ is defined in Lemma~\ref{l:remainderEstimate} and 
$\bar\varphi(x) \equiv \E^{-x} \Nn(d_{2}(-x,I(x)))$.
By the arithmetic-geometric inequality, $d_{2}(-x,I(x))= -\frac{|x|}{I(x)\sqrt{T}}-\frac{I(x)\sqrt{T}}{2} \le -\sqrt{2|x|}$, hence $0 \le \bar\varphi(x) \le \E^{-x} \Nn(-\sqrt{2|x|})$.
The expansion
$\Nn(z) = n(z)\left(-z^{-1} + \Oo\left(z^{-3}\right)\right)$,
as $z$ tends to $-\infty$ yields $\E^{-x} \Nn(-\sqrt{2|x|}) = \frac1{\sqrt{2\pi}\sqrt{2|x|}} + \Oo\left(|x|^{-3/2}\right)$, and therefore
\be \label{e:barVarphiEstimate}
\limsup_{x \downarrow -\infty}\left(\bar\varphi(x) \sqrt{2\pi}\sqrt{2|x|}\right) \le 1.
\ee
Now using $d_2(x, I(x))=-d_1(-x,I(x))$, it follows from~\eqref{e:d1parityEstimate} 
and Lemma~\ref{l:remainderEstimate} that $\lim_{x \downarrow -\infty}\Nn(-d_{2}(x,I(x))) = 0$ when $\pp=0$, and hence $d_{2}(x,I(x))$ diverges to $+\infty$.
If $\pp\neq 0$, we get
\[
\begin{aligned}
d_2(x, I(x)) &= -\Nn^{-1} \left(\pp + R(K_x) + \bar\varphi(x) \right)
\\
&= -\Nn^{-1}(\pp) - n(\Nn^{-1}(\pp))^{-1} \left[R(K_x) + \bar\varphi(x)\right]
+ \Oo(R(K_x)^2) + \Oo(\bar\varphi(x)^2)
\\
&= -\qq - \E^{\qq^2/2} \sqrt{2\pi}\bar\varphi(x) - \E^{\qq^2/2} \sqrt{2\pi} R(K_x)
+ \Oo\left((F(K_x)-F(0))^2\right) + \Oo\left(|x|^{-1}\right),
\end{aligned}
\]
where we have used the bound on $R$ in Lemma~\ref{l:remainderEstimate} 
and estimate~\eqref{e:barVarphiEstimate} in the last step.
The claimed estimate is obtained setting $\varphi \equiv \sqrt{2\pi}\bar\varphi$ and using the expression of~$R$ in Lemma~\ref{l:remainderEstimate}.
\end{proof}
\medskip

\begin{proof}[Proof of Theorem~\ref{t:secondOrder}]
We first prove the limit~\eqref{e:secondOrderNoAtom}.
Assume that $\pp=0$.
Estimate~\eqref{e:d1parityEstimate} implies that for every $M>0$ we have $d_1(-x,I(x)) = \frac{x}{I(x)\sqrt{T}}+\frac{I(x)\sqrt{T}}{2} < -M$ for $x$ small enough, or yet $I(x)\sqrt{T} < -M +\sqrt{M^2+2|x|}$.
Therefore~\eqref{e:secondOrderNoAtom} follows from the following, which holds for every $M>0$:
\[
\limsup_{x \downarrow -\infty}
\left(I(x)\sqrt{T}-\sqrt{2|x|}\right) < -M +\limsup_{x \downarrow -\infty}(\sqrt{M^2+2|x|}-\sqrt{2|x|})
=-M.
\]

We now move on to the proof of~\eqref{e:firstExpansion}.
Let us write $d_2$ instead of $d_2(x, I(x))$ for simplicity.
According to the definition of $d_2$ on Page~\pageref{eq:d12}, $I(x)$ satisfies
\[
I(x) = \frac{1}{\sqrt T}\left(-d_2+\sqrt{d_2^2-2x}\right) 
 = \frac{2|x|}{\sqrt{2|x|T + T d_2^2} + \sqrt T d_2},
\]
so that 
\be \label{e:AandB}
\sqrt{T} I(x) - \sqrt{2|x|} = \frac{2|x| - \sqrt{4 x^2 + 2|x| d_2^2 }}{\sqrt{2|x| + d_2^2} + d_2}
-
\frac{\sqrt {2|x|} d_2}{\sqrt{2|x| + d_2^2} + d_2}
=: A(x) - B(x).
\ee
The expansion $\sqrt{f(x)} - \sqrt{f(x) + g(x)} = -\frac{g(x)}{2\sqrt{f(x)}} + \Oo(\frac{g(x)^2}{f(x)^{3/2}})$ 
when $f(x) \uparrow \infty$ 
and $g = o(f)$, together with
\be \label{e:factorExp}
\Bigl(\sqrt{2|x| + d_2^2} + d_2\Bigr)^{-1}=
\frac1{\sqrt{2|x|}} \left(1-\frac{d_2}{2\sqrt{2|x|}} + \Oo\left(|x|^{-1}\right)\right),
\qquad \text{as } x \downarrow-\infty,
\ee
allow to see that, as $x$ tends to $-\infty$,
\be \label{e:Aasympt} 
A(x) = \left(-\frac{d_2^2}2 + \Oo\left(|x|^{-1}\right)\right)
\frac1{\sqrt{2|x|}} \left(1-\frac{d_2}{2\sqrt{2|x|}} + \Oo\left(|x|^{-1}\right)\right)
=
-\frac{d_2^2}{2\sqrt{2|x|}} + \Oo\left(\frac{1}{|x|}\right).
\ee
On the other hand,
$\lim\limits_{x \downarrow -\infty} B(x) = \lim\limits_{x \downarrow -\infty} d_2 = -\qq$; 
then consider $B(x) + \qq = B_1(x)+B_2(x)$,
where
\[
B_1(x) \equiv
\frac{\qq d_2}{\sqrt{2|x| + d_2^2} + d_2},
\qquad
B_2(x) \equiv h(x) \left( d_2 + \qq \sqrt{1 + \frac{d_2^2}{2|x|}} \right),
\]
and $h(x) \equiv \frac{\sqrt{2|x|}} {\sqrt{2|x| + d_2^2} + d_2}$.
Using~\eqref{e:factorExp}, one has $B_1(x) = \frac{\qq d_2}{\sqrt{2|x|}}+\Oo(|x|^{-1})$ 
as $x \downarrow -\infty$.
Moreover, since
\[
d_2 + \qq \sqrt{1 + \frac{d_2^2}{2|x|}} =
-\frac{\qq d_2^2}{4|x|} + \tilde \varphi(x)+\Oo(|x|^{-2}),
\]
where $\tilde{\varphi}(x)\equiv d_2 + \qq = \Oo(|x|^{-1/2})$ by Lemma~\ref{l:d2Estimate}, we obtain
$$
B_2(x) = 
\left(1-\frac{d_2}{2\sqrt{2|x|}} + \Oo\left(|x|^{-1}\right)\right)
\left(-\frac{\qq d_2^2}{4|x|} + \tilde \varphi(x)+\Oo(|x|^{-2}) \right)
=
\tilde \varphi(x) - \frac{\qq d_2^2}{4|x|}
+\Oo\left(\tilde \varphi(x) |x|^{-1/2}\right).
$$
Finally putting~\eqref{e:Aasympt}, the estimate on $B_1(x)$ and this estimate on $B_2(x)$
together, it follows from~\eqref{e:AandB} that
\[
\begin{aligned}
\sqrt{T} I(x) - \sqrt{2|x|} - \qq
&=
A(x)-(B_1(x)+B_2(x))
\\
&= -\frac{d_2}{\sqrt{2|x|}} \left(\frac{d_2}2 + \qq\right)
-\tilde\varphi(x)
+\Oo\left(\tilde \varphi(x)|x|^{-1/2}\right)
+ \Oo(|x|^{-1})
\\
&= \frac{\qq - \tilde\varphi(x)}{\sqrt{2|x|}} \left(\frac{\qq}{2} + \frac{\tilde\varphi(x)}2\right)
-\tilde\varphi(x)
+\Oo\left(\tilde \varphi(x)|x|^{-1/2}\right)
+ \Oo(|x|^{-1})
\\
&= \frac{\qq^2}{2{\sqrt{2|x|}}} 
-\tilde\varphi(x)
+ \Oo(|x|^{-1}),
\end{aligned}
\]
as $x \downarrow -\infty$.
The claimed estimate on $I(x)$ now follows from Lemma~\ref{l:d2Estimate}, setting $\overline\chi(x) \equiv \frac{\qq^2}2 - \sqrt{2|x|} \tilde\varphi(x) =
\frac{\qq^2}2 + \E^{\qq^2/2}\sqrt{2|x|} \varphi(x)$, where $\varphi$ is given in Lemma~\ref{l:d2Estimate}.
\end{proof}

\section{A refined formula for the implied volatility} \label{s:asymFormAtom}

The estimates~\eqref{e:secondOrderExplicitError} and~\eqref{e:secondOrderWithErrorEstimate} contain a global $\Oo(|x|^{-1/2})$ error term.
Slightly after the first version of this paper appeared, Gulisashvili~\cite{ArchilAtom} proved a similar, albeit slightly different, expansion for the left wing of the smile when the stock price has a mass at the origin, 
with a stronger $\Oo(|x|^{-3/2})$ term.
Under the assumption
\be \label{e:archilCdf}
F(K) - F(0) = \Oo\left(|\log(K/S_0)|^{-3/2}\right)
\qquad
\mbox{as } K \downarrow 0,
\ee
Gulisashvili~\cite[Corollary 8]{ArchilAtom} proves the expansion
\be \label{e:archilAtom}
\sqrt T I(x) =
\sqrt{2 |x|}
+ U(x,\cdot)^{-1}(\pp)
+ \frac{(U(x,\cdot)^{-1}(\pp))^2}{2\sqrt{2|x|}}
+ \Oo\left(|x|^{-3/2}\right),
\qquad \mbox{as $x \downarrow -\infty$},
\ee
where the function $U$ is defined by $U(x,z) \equiv \Nn(z) - n(z)/\sqrt{2|x|}$.
It is easy to check that, for every $x \neq 0$, the inverse function $U(x, \cdot)^{-1}$ is well-defined on the interval $(0,1)$.
In~\cite[Theorem 6]{ArchilAtom}, an expansion analogous to~\eqref{e:archilAtom} is first proved with 
$U(x,\cdot)^{-1}(\pp)$ replaced by $U(x,\cdot)^{-1}(P(K_x)/K_x)$, without assuming~\eqref{e:archilCdf}.
Formula~\eqref{e:archilAtom} then follows as a corollary under Condition~\eqref{e:archilCdf}.
The main difference between~\eqref{e:archilAtom} and the lower-order 
expansions~\eqref{e:secondOrderExplicitError}-\eqref{e:secondOrderWithErrorEstimate} 
is that the $x$-dependence of the term $U(x,\cdot)^{-1}(\pp)$ is not explicit.
We are going to refine~\eqref{e:archilAtom} in two directions:
\begin{itemize}
\item[1.] we provide an explicit asymptotic formula~\eqref{e:explicitExp} 
with the same accuracy as~\eqref{e:archilAtom};
by `explicit', we refer to the fact that our final expression is an expansion in powers of $\sqrt{|x|}$, 
with explicit coefficients.
\item[2.] we estimate the constant in front of the $\Oo\left(|x|^{-3/2}\right)$ error term.
\end{itemize}

The following result refines~\cite[Theorem 6]{ArchilAtom}.

\begin{theorem}\label{thm:ExpansionIV}
Assume $\pp \in (0,1)$. Recall the function $R(\cdot)$ defined in Lemma~\ref{l:remainderEstimate} and set $r(x) = R(K_x)$, so that $r(x) = \frac1 {K_x} \int_0^{K_x} (F(y)-F(0)) dy$.
The following expansion holds for any $x<0$:
\[
I(x)\sqrt{T} =
\sqrt{2 |x|}
 + U(x,\cdot)^{-1}(\pp + r(x))
 + \frac{U(x,\cdot)^{-1}(\pp + r(x))^2}{2\sqrt{2|x|}}
 + \left(a(x) - \frac{U(x,\cdot)^{-1}(\pp + r(x))^4}{16\sqrt{2}}\right)\frac{1}{|x|^{3/2}}
 + o\left(|x|^{-3/2}\right),
\]
where the function~$a(\cdot)$ is such that $a(x) \le 0$ for $x$ small enough, 
and $\liminf\limits_{x \downarrow -\infty} a(x) \ge -A(\qq) \equiv -\displaystyle \frac{3\qq^2+2}{4\sqrt{2}} e^{\qq^2/2}$.
\end{theorem}
\begin{proof}
Set $u(x) \equiv U(x,\cdot)^{-1}(\pp + r(x))$ and $u_{\eps}(x) \equiv U(x,\cdot)^{-1} \left( \pp + r(x) - \frac {B_{\eps}}{|x|^{3/2}} \right)$, where $B_{\eps} = \frac{3\qq^2+2+\eps}{8 \sqrt{\pi}}$.
The following estimates are given in~\cite[Theorem 12]{ArchilAtom}: for every $\eps \maj 0$, there exist $x_{\eps}$ such that\footnote{En passant, we correct a sign error in the definition of~$u_{\eps}$
in~\cite[Equation (2.5)]{ArchilAtom}, 
where $u_{\eps}(x) \equiv U(x,\cdot)^{-1} \left( \pp + r(x) + \frac{B_{\eps}}{|x|^{3/2}} \right)$.
Since $U(x,\cdot)^{-1}$ is increasing on $(0,1)$, the latter definition entails $u_{\eps}(x) \maj u(x)$ for $x$ small.
This would imply $h_{2,\eps}(x) \maj h_1(x)$, in contradiction with~\eqref{e:implVolUpperAndLower}.
The correct definition of~$u_{\eps}$ follows by inspection of the proof of~\cite[Theorem 12]{ArchilAtom}.}
\be \label{e:implVolUpperAndLower}
\sqrt 2 \sqrt{|x| + h_{2,\eps}(x)}
\le \sqrt T I(x) \le
\sqrt 2 \sqrt{|x| + h_1(x)},
\ee
for all $x \mino x_{\eps}$.
The functions $h_1$ and $h_{2,\eps}$ in~\eqref{e:implVolUpperAndLower} are defined as follows:
\[
h_1(x) = u(x) \sqrt{2|x| + u(x)^2} + u(x)^2,
\qquad
h_{2,\eps} = u_{\eps}(x) \sqrt{2|x| + u_{\eps}(x)^2} + u_{\eps}(x)^2
\]
The estimates~\eqref{e:implVolUpperAndLower} are proved in~\cite{ArchilAtom} by comparing the Put price $P_{\BS}(x,I(x))$ with two properly chosen Put prices $P_{\BS}(x,I_1(x))$ and $P_{\BS}(x,I_2(x))$.

\emph{Step 1} (Estimate of the difference between the upper and lower bounds 
in~\eqref{e:implVolUpperAndLower}). We first estimate the difference between $u(x)$ and $u_{\eps}(x)$.
By the mean value theorem, there exists $\xi\in(u(x), u_\eps(x))$ such that
$U(x,u(x)) - U(x,u_\eps(x)) = \partial_z U(x, \xi)(u(x) - u_{\eps}(x))$, where $\partial_z U(x, z) = n(z)\left(1 + \frac{z}{\sqrt{2|x|}} \right)$.
It is not hard to see that $U(x,\cdot)^{-1}(\pp) \to \mathcal{N}^{-1}(\pp) = \qq$ as $x \downarrow -\infty$.
Next, we observe that the function $U(x,\cdot)^{-1}$ is locally Lipschitz inside $(0,1)$, with a Lipschitz constant independent of $x$ on any fixed subinterval of $(0,1)$.
Since $r(x)$ converges to zero, it follows that $u(x)$ tends to~$\qq$ and~$u_{\eps}(x)$ to~$\qq$ 
as $x \downarrow -\infty$.
Overall, $\xi$ converges to~$\qq$, so that
\[
u(x) - u_{\eps}(x)
=
\frac{1}{\partial_z U(x, \xi)}
\frac{B_\eps}{|x|^{3/2}}
\sim
\frac{1}{n(\qq)} \frac{B_\eps}{|x|^{3/2}}
= 
\frac{A_{\eps}}{|x|^{3/2}}
\qquad
\mbox{as } x \downarrow -\infty,
\]
where $A_\eps = \frac{3\qq^2+2+\eps}{4\sqrt{2}} e^{\qq^2/2}$.

We now estimate the difference between the functions $h_1$ and $h_{2,\eps}$.
We have already proved above that both~$u$ and~$u_\eps$ are bounded at $-\infty$. 
Using the Taylor expansion $\sqrt{1+y} = 1+\frac12 y + \Oo(y^2)$ as $y \to 0$, we have
\[
\begin{aligned}
h_1(x) - h_{2,\eps}(x)
&=
\sqrt{2|x|} \left( u(x) \sqrt{1 + \frac{u(x)^2}{2|x|}} - u_{\eps}(x) \sqrt{1 + \frac{u_\eps(x)^2}{2|x|}} \right)
+ u(x)^2 - u_{\eps}(x)^2
\\
&=
\sqrt{2|x|} \left( u(x) - u_{\eps}(x) 
+ \frac{u(x)^3 - u_{\eps}(x)^3}{4|x|}
+ \Oo(|x|^{-2}) \right)
+ \Oo\left(u(x) - u_{\eps}(x)\right)
\\
&=
\frac{\sqrt 2 A_{\eps}(\qq)}{|x|} (1+o(1)).
\end{aligned}
\]
Finally, we can estimate the difference between the upper and lower bounds in~\eqref{e:implVolUpperAndLower}.
Since the functions $h_1$ and $h_{2,\eps}$ are of order $\sqrt{|x|}$ when $x \downarrow -\infty$, we get
\be \label{e:differenceH1H2}
\begin{aligned}
\sqrt 2 \sqrt{|x| + h_1(x)}
-
\sqrt 2 \sqrt{|x| + h_{2,\eps}(x)}
&=
\sqrt 2
\frac{h_1(x) - h_{2,\eps}(x)}
{\sqrt{|x| + h_1(x)}+\sqrt{|x| + h_{2,\eps}(x)} }
\\
&\sim
\frac{h_1(x) - h_{2,\eps}(x)}
{\sqrt{2|x|}}
\sim 
\frac{A_{\eps}}{|x|^{3/2}}
\qquad
\mbox{as } x \downarrow -\infty.
\end{aligned}
\ee

\emph{Step 2} (Final estimate of the implied volatility).
Using twice the Taylor expansion $\sqrt{1+y} = 1+\frac12 y -\frac18 y^2 + \Oo(y^3)$ for small~$y$, 
it follows from the definition of the function~$h_1$ that 
\be \label{e:expUpperBound}
\sqrt{|x| + h_1(x)}
= \sqrt{|x|} + \frac{u(x)}{\sqrt 2} 
+ \frac{u(x)^2}{4 \sqrt{|x|}}
- \frac{u(x)^4}{32 |x|^{3/2}}
+ o\left(|x|^{-3/2}\right)
\ee
(the term containing $u(x)^3$ disappears because of some cancellations).
Combining~\eqref{e:implVolUpperAndLower} and~\eqref{e:expUpperBound}, we obtain
\[
\sqrt T I(x) = 
\sqrt{|x|} + u(x)
+ \frac{u(x)^2}{2 \sqrt 2 \sqrt{|x|}}
- \frac{u(x)^4}{16 \sqrt 2 |x|^{3/2}}
+ \alpha(x)
+ o\left(|x|^{-3/2}\right),
\]
where the function $\alpha$ satisfies the following estimate: for every $\eps \maj 0$, $-(\sqrt 2 \sqrt{|x| + h_1(x)} - \sqrt 2 \sqrt{|x| + h_{2,\eps}(x)}) \le \alpha(x) \le 0$ for all $x \mino x_{\eps}$.
The final claim then follows by setting $a(x) = \alpha(x) |x|^{3/2}$ and using~\eqref{e:differenceH1H2}.
\end{proof}
\medskip

Define the function
\be \label{e:explicitExp}
\widetilde{I}(x) \equiv \sqrt{2 |x|}
 + \qq + \frac{\qq^2 + 2}{2\sqrt{2|x|}}
 + \frac{\qq}{4|x|},
\qquad \text{for }x < 0.
\ee

\begin{theorem} \label{t:IVArchil}
Assume $\pp \maj 0$ and Condition~\eqref{e:archilCdf}.
\begin{itemize}
\item[(i)] The implied volatility satisfies $\sqrt T I(x) = \widetilde I(x) + \Oo(|x|^{-3/2})$ as $x$ tends to $-\infty$;
\item[(ii)] if moreover 
\be \label{e:cdfLittleO}
F(K) - F(0) = o\left(|\log(K/S_0)|^{-3/2}\right),
\ee
for small~$K$, then the implied volatility satisfies
\begin{align} \label{eq:IVExpansion_up}
&\limsup_{x\downarrow-\infty}|x|^{3/2}\left(
\sqrt{T} I(x) - \widetilde{I}(x),
\right) \leq \mathfrak{c_q}
\\ \label{eq:IVExpansion_low}
&\liminf_{x\downarrow-\infty}|x|^{3/2}\left( \sqrt{T} I(x) - \widetilde{I}(x)\right) \geq 
\mathfrak{c_q} - A(\qq),
\end{align}
where $\mathfrak{c_q} = \frac{1-\qq^2}{12\sqrt{2}} - \frac{\qq^4}{16\sqrt{2}}$ and the constant $A(\qq)$ is defined in Theorem~\ref{thm:ExpansionIV}.

Finally, if there exists $\alpha>0$ such that
\be \label{e:cdfExactAsympt}
F(K) - F(0) \sim \alpha |\log(K/S_0)|^{-3/2},
\ee
as $K$ tends to zero, then~\eqref{eq:IVExpansion_up} and~\eqref{eq:IVExpansion_low} 
hold with $\mathfrak{c_q}$ replaced by $\mathfrak{c_q} + \alpha\sqrt{2 \pi}\E^{\qq^2/2}$.
\end{itemize}
\end{theorem}
\medskip

Theorem~\ref{t:IVArchil}(i) provides an explicit formula for the implied volatility with a $\Oo(|x|^{-3/2})$ 
error term, and~(ii) estimates this error term from above and below (under the slightly stronger 
condition~\eqref{e:cdfLittleO} or~\eqref{e:cdfExactAsympt}).
In particular, the following remarks are in order:

\begin{itemize}
\item 
Definition~\eqref{e:explicitExp} highlights the presence in the expansion of a term proportional to $|x|^{-1}$, which is hidden in Gulisashvili's formulation~\eqref{e:archilAtom}.
From the point of view of numerical evaluation, the terms in~\eqref{e:explicitExp} are elementary functions of the log-strike~$x$,  while the evaluation of~\eqref{e:archilAtom} requires to numerically invert the function $U(x,\cdot)$ for each value of~$x$.

\item If the underlying stock price is distributed according to the measure
\be \label{e:measure}
\mu(\D K) = \pp \delta_0(\D K) + (1-\pp) f(K) \D K,
\ee
where $f$ is some pdf on $(0,\infty)$ such that $f(K) = \Oo(K^{-a})$ for some $a<1$ as $K \downarrow 0$, 
then it is immediate that $ F(K)-F(0) = \Oo(K^{1-a})$, therefore Condition~\eqref{e:cdfLittleO} is fulfilled.\til
(Notice that, if $f(K) \sim c K^{-a}$ as $K\downarrow 0$ for some $c>0$, the restriction $a < 1$ is trivially necessary to ensure integrability).
Most of the financial models (with mass at zero) used in practice satisfy~\eqref{e:measure} with $f(K) = \Oo(K^{-a})$; in particular, the Merton model with jump-to-default in Section~\ref{s:merton}
(where the density $f$ tends to zero at the origin), and the CEV model in Section~\ref{s:cev} (where $f$ explodes at the origin). 
\end{itemize}

The constants $\mathfrak{c_q}$ and $A(\qq)$ in Theorem~\ref{t:IVArchil} are functions of $\qq^2 = (\mathcal{N}^{-1}(\pp))^2$, therefore--in terms of~$\pp$--they are symmetric around $\pp=1/2$. 
The function $A(\cdot)$ attains its minimum value for $\qq=0$ (or yet $\pp=1/2$), 
so that the bounds are the tightest at this point.
This is illustrated in Figure~\ref{f:upperLowerBoundsErrorTerm}.

\begin{figure}[t]
     \begin{center}
            \includegraphics[scale=0.4]{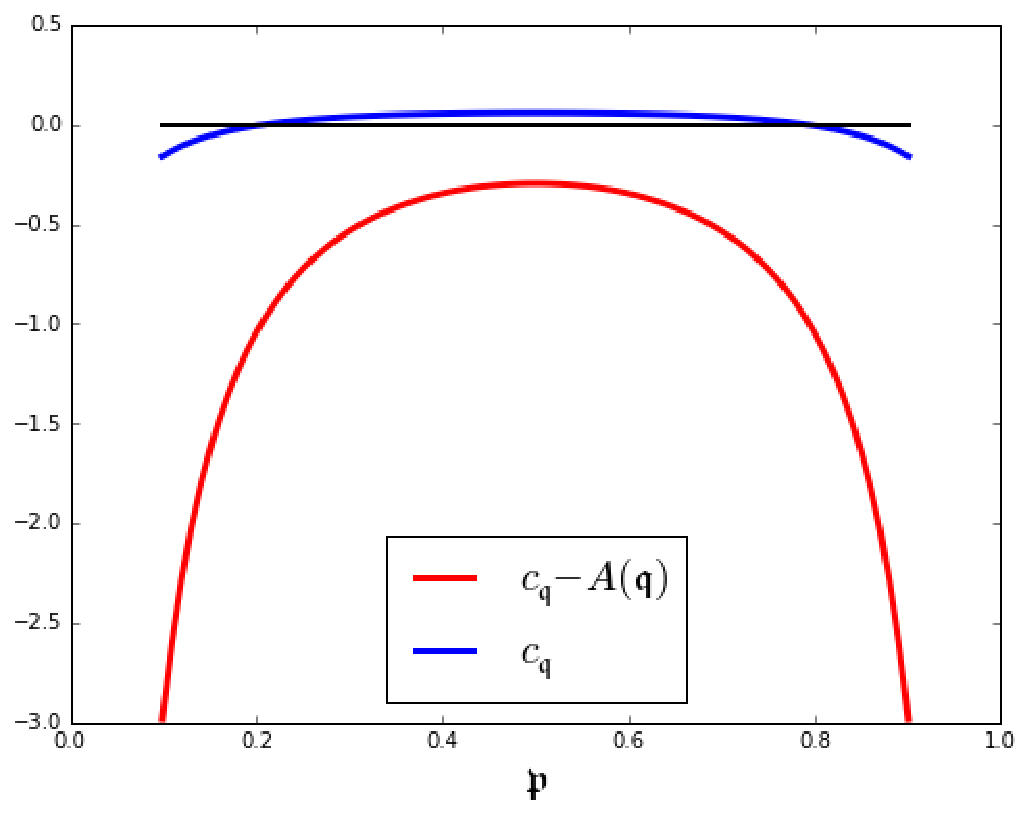}
    \end{center}
\caption{The upper and lower bounds $\mathfrak{c_q}$ and $\mathfrak{c_q}-A(\qq)$ in Theorem~\ref{t:IVArchil}, plotted in terms of the mass at zero $\pp = \mathcal{N}(\qq)$.}
\label{f:upperLowerBoundsErrorTerm}
\end{figure}

\begin{proof}
Consider the function $V(y,z) \equiv \Nn(z) - y n(z)$.
It is not difficult to see that, for every $y \in \R$, the equation $V(y,z) = \pp$ has a unique solution $z = \widetilde z(y)$ (take derivatives with respect to~$z$). 
In particular, $\widetilde z(0) = \Nn^{-1}(\pp) = \qq$. 
By the implicit function theorem, the function $\widetilde z$ is infinitely differentiable,
and we can easily compute its derivatives by implicit differentiation.
For the first derivative, $\widetilde{z}'(y) = (1 + y \widetilde z(y))^{-1}$,
which converges to~$1$ as~$y$ tends to zero;
iterating the procedure, we obtain the Taylor expansion
\[
\widetilde{z}(y) = \qq + y - \frac{\qq}2y^2 + 
\frac13 (\qq^2-1) y^3 + \Oo\left(|y|^4\right),
\qquad
\mbox{as } y \text{ tends to zero}.
\]
Since, for every $x \mino 0$, $U(x,\cdot)^{-1}(\pp) = \widetilde{z}(1/\sqrt{2|x|})$, we obtain the following expansion as~$x$ tends to $-\infty$:
\be \label{e:Uexpansion}
U(x,\cdot)^{-1}(\pp) = \qq + \frac{1}{\sqrt{2|x|}} - \frac{\qq}{4|x|}
 + \frac{\sqrt{2}}{12}\frac{\qq^2-1}{|x|^{3/2}} + \Oo\left(|x|^{-2}\right).
\ee
We have to estimate the difference $U(x,\cdot)^{-1}(\pp  + r(x)) - U(x,\cdot)^{-1}(\pp)$ as $x \downarrow -\infty$.
Applying the mean value theorem to the function $U(x,\cdot)$ and proceeding 
as in the proof of Theorem~\ref{thm:ExpansionIV}, we obtain
\be \label{e:Uestimate}
U(x,\cdot)^{-1}\left(\pp  + r(x) \right) - U(x,\cdot)^{-1}(\pp)
\sim
\frac{r(x)}{\phi(\qq)} = \sqrt{2 \pi} e^{\qq^2/2} r(x),
\qquad
\mbox{as } x \downarrow -\infty.
\ee
Using Lemma~\ref{l:remainderEstimate} and Condition~\eqref{e:archilCdf}, we have $r(x) = \Oo(|x|^{-3/2})$.
Now, plugging~\eqref{e:Uestimate} and~\eqref{e:Uexpansion} into Theorem~\ref{thm:ExpansionIV} and collecting terms of the same order in~$x$, we get
\[
\sqrt{T} I(x) = \widetilde{I}(x) + (\mathfrak{c_q} + a(x)) \frac{1}{|x|^{3/2}}
+ \sqrt{2 \pi} \E^{\qq^2/2} r(x) \left(1+o(1)\right)
+ \Oo \left(|x|^{-2}\right).
\]
Using the boundedness of the function $a(\cdot)$ at $-\infty$, we obtain~(i) under Condition~\eqref{e:archilCdf}.
If moreover Condition~\eqref{e:cdfLittleO} holds, the bounds~\eqref{eq:IVExpansion_up} 
and~\eqref{eq:IVExpansion_low} follow from Theorem~\ref{thm:ExpansionIV} 
and Lemma~\ref{l:remainderEstimate}.
Finally, when Condition~\eqref{e:cdfExactAsympt} holds, the last statement follows from $r(x) = \frac1{K} \int_0^{K} (F(y)-F(0))dy \sim F(K)-F(0) \sim \alpha |\log(K/S_0)|^{-3/2} = \alpha |x|^{-3/2}$ 
as~$x$ tends to~$-\infty$.
\end{proof}

\begin{remark}[\textbf{Asymptotic shapes of implied volatility in stochastic volatility models}] \label{r:stochVol}
The asymptotic expansion `leading-order term proportional to $\sqrt{|x|}$ + constant + vanishing term' 
is typical in stochastic volatility models.
Yet, in this case the phenomenon has a different nature: when the stock price follows an exponential (hence strictly positive) diffusion process with stochastic volatility, the functional form of the implied volatility at low strikes is rather determined by the asymptotics of the density of the asset price close to zero. 
In Theorem~\ref{t:IVArchil}, the same parametric form relates to the presence of an atom at zero, but is independent from the behaviour of the remaining distribution on $(0,\infty)$
(as soon as Condition~\eqref{e:archilCdf} is in force).
Some examples are:

\begin{enumerate}
\item the (uncorrelated) Stein-Stein model~\cite{SteinStein}. Gulisashvili and Stein~\cite[Theorem 3.1]{GulSt}
prove the following expansion for the implied volatility:
\[
\sqrt{T} I(x) = \gamma_1 \sqrt{x} + \gamma_2 +\Oo\bigl(|x|^{-1/2}\bigr)
\quad \mbox{as $x \uparrow +\infty$},
\]
where $\gamma_1 \in (0,\sqrt{2})$ and $\gamma_2 >0$ are constants that depend on the model parameters.
Since in uncorrelated volatility models the smile is symmetric, see~\cite{RenTouz}, 
the same expansion holds when~$x$ tends to~$-\infty$;

\item the Heston model~\cite{Hest}, for which Friz, Gerhold, Gulisashvili and Sturm~\cite[(4.11)]{Refined} 
prove the expansion:
\[
\sqrt{T} I(x)
= \rho_1 \sqrt{|x|} + \rho_2 + \rho_3 \frac{\log(|x|)}{\sqrt{|x|}}
+\Oo\bigl(|x|^{-1/2}\bigr),
\qquad\text{as }x\downarrow-\infty,
\]
where the coefficients $\rho_1 \in (0,\sqrt2)$, $\rho_2$ and~$\rho_3$ are related to the model parameters.
One can notice the appearance of the function $\log|x|$ in the third-order term;

\item the uncorrelated Hull-White model (Example~\ref{ex:HullWhite}), 
for which Gulisashvili and Stein~\cite{GulStHW} prove
\[
\sqrt{T} I(x) = \sqrt{2 |x|} - \frac{\log |x| + \log \log |x| }{2T\xi} + \Oo(1), \qquad\text{as } x\downarrow-\infty.
\]
The constant second-order term appearing in the previous expansions is replaced here by a term 
diverging to~$-\infty$, in agreement with Theorem~\ref{t:secondOrder}.
\end{enumerate}
\end{remark}


\section{Examples and numerics} \label{s:examples}

A distribution~$\mu$ with a mass $\pp \in (0,1)$ at zero can be written in terms of its Jordan decomposition
\be \label{e:distribAtom}
\mu(\D s) = \pp \delta_0(\D s) + (1-\pp) \tilde \mu(\D s),
\ee
where $\tilde \mu$ is a probability measure on $(0,\infty)$.
The martingale condition $\esp[S_T]=\int_{[0,\infty)} s \mu(\D s)=S_0$ imposes the constraint
$\int_{(0,\infty)} s \tilde \mu(\D s) = S_0/(1-\pp)$.
In order to illustrate Theorem~\ref{t:IVArchil}, we compute and plot the function
$J(x) \equiv I(x) \sqrt{\frac{T}{|x|}}$,
which must tend to $\sqrt2$ as $x \downarrow -\infty$, 
and compare it to the expansion~$\tilde J$ given by Theorem~\ref{t:IVArchil}:
\be \label{e:expansionPlot}
\tilde J(x) \equiv \displaystyle \frac1{\sqrt{|x|}} \left(
\sqrt{2|x|} + \qq + \frac{\qq^2 + 2}{2\sqrt{2|x|}}
+ \frac{\qq}{4|x|} \right)
\ee

\subsection{A toy example}\label{ex:affine}
We define a piecewise affine Call price on $[0,\infty)$ by setting
\be \label{e:affineCall}
\widetilde{C}(K) = (S_0 - (1-\pp)K)^+, \qquad K \ge 0.
\ee
The corresponding asset price distribution has the form~\eqref{e:distribAtom}, with $\tilde \mu(ds)=\delta_{S_0/(1-\pp)}(\D s)$.
The cumulative distribution of $\mu$, $F(K) = \mu([0,K]) = \pp + (1-\pp) \ind_{\{K \ge S_0/(1-\pp)\}}$, 
is constant for $K<S_0/(1-\pp)$, hence Condition~\eqref{e:archilCdf} is trivially satisfied.
Figure~\ref{f:smilesAffine} shows some numerical results for $T=1$.

\begin{figure}[!t]
     \begin{center}
        \subfigure[Normalised smile from affine Call price, $\pp=0.1$]{
        \label{fig:AffinePriceApproxP=0.1}
            \includegraphics[width=0.45\textwidth]{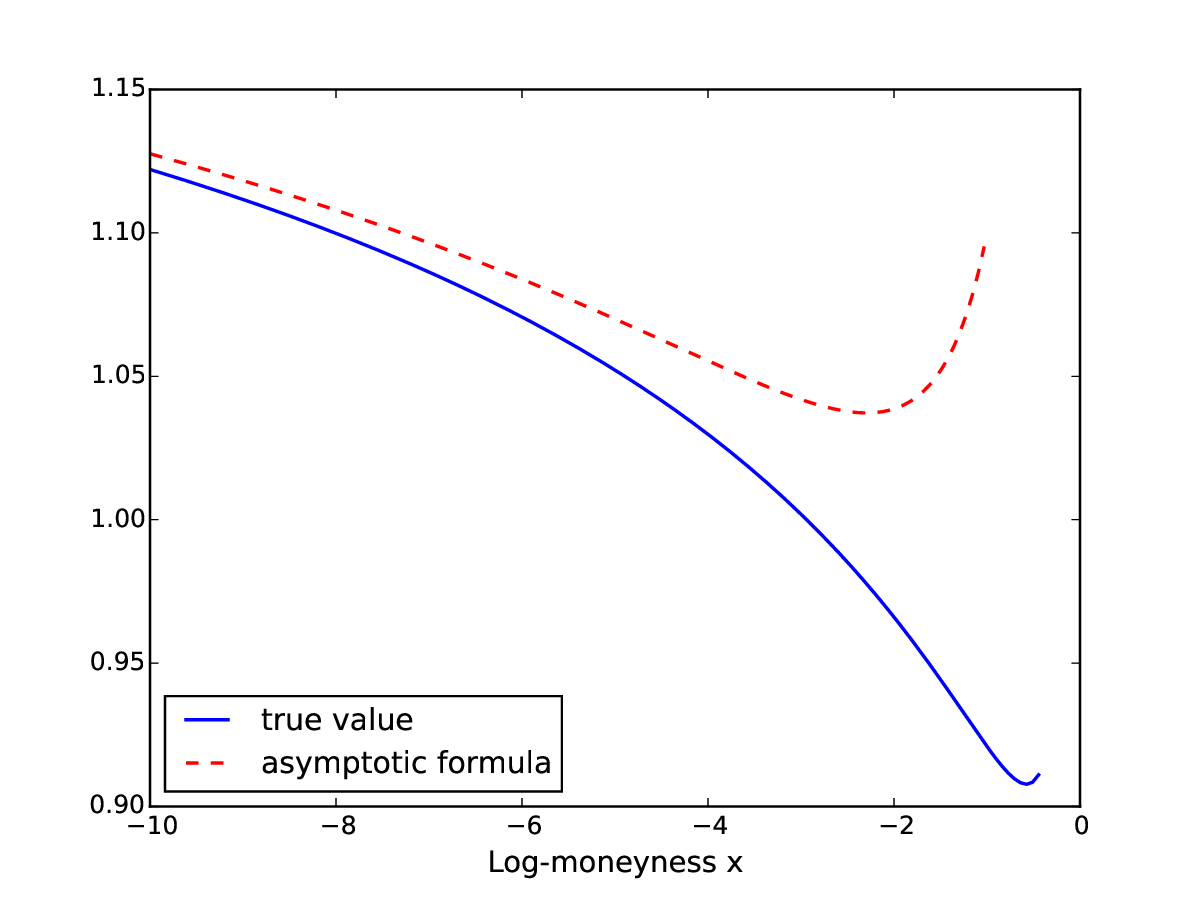}
            \graphicspath{ {./pics} }
        }
        \subfigure[Normalised smile from affine Call price, $\pp=0.9$]{
        \label{fig:AffinePriceApproxP=0.9}
           \includegraphics[width=0.45\textwidth]{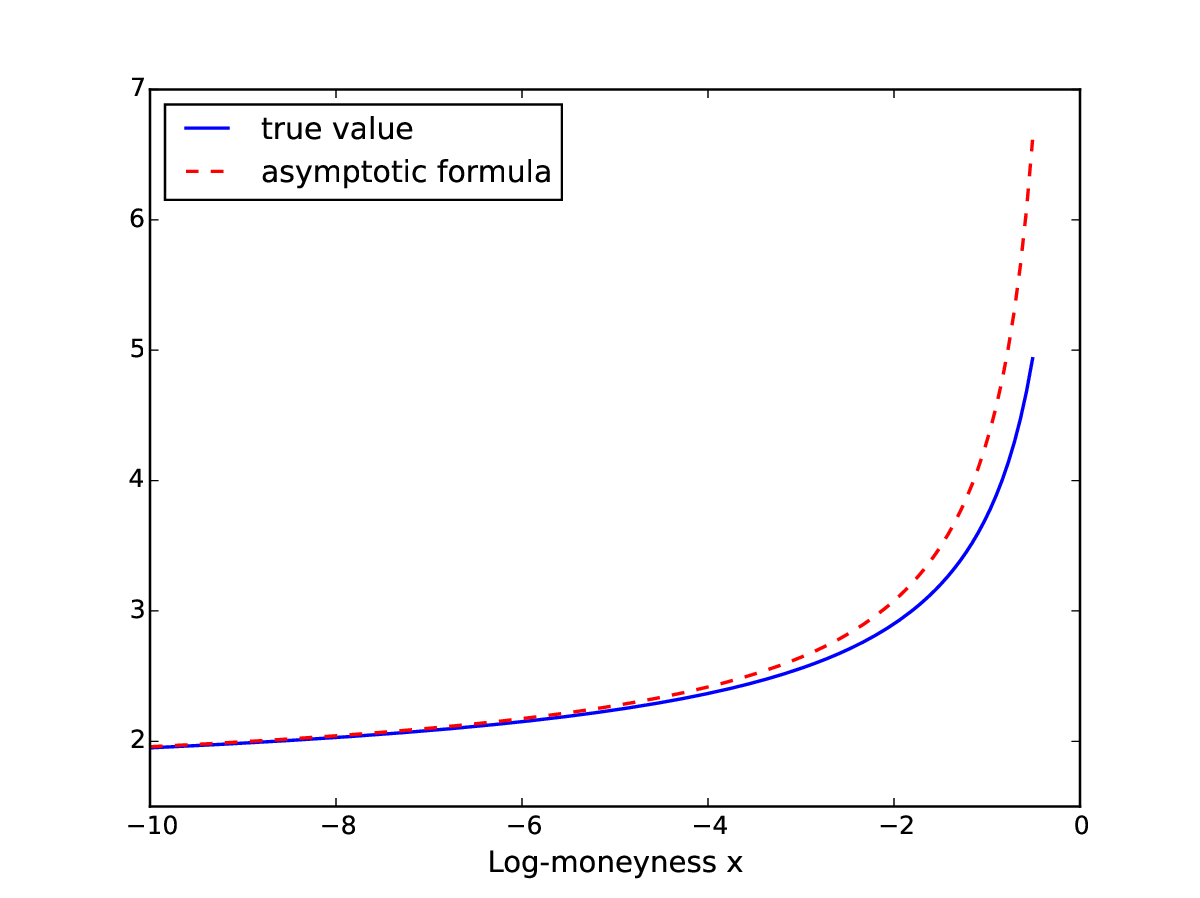}
           \graphicspath{ {./pics} }
        }\\ 
        \subfigure[Normalised smile from affine Call price, $\pp=0.5$]{
        \label{fig:AffinePriceApproxP=0.5}
           \includegraphics[width=0.45\textwidth]{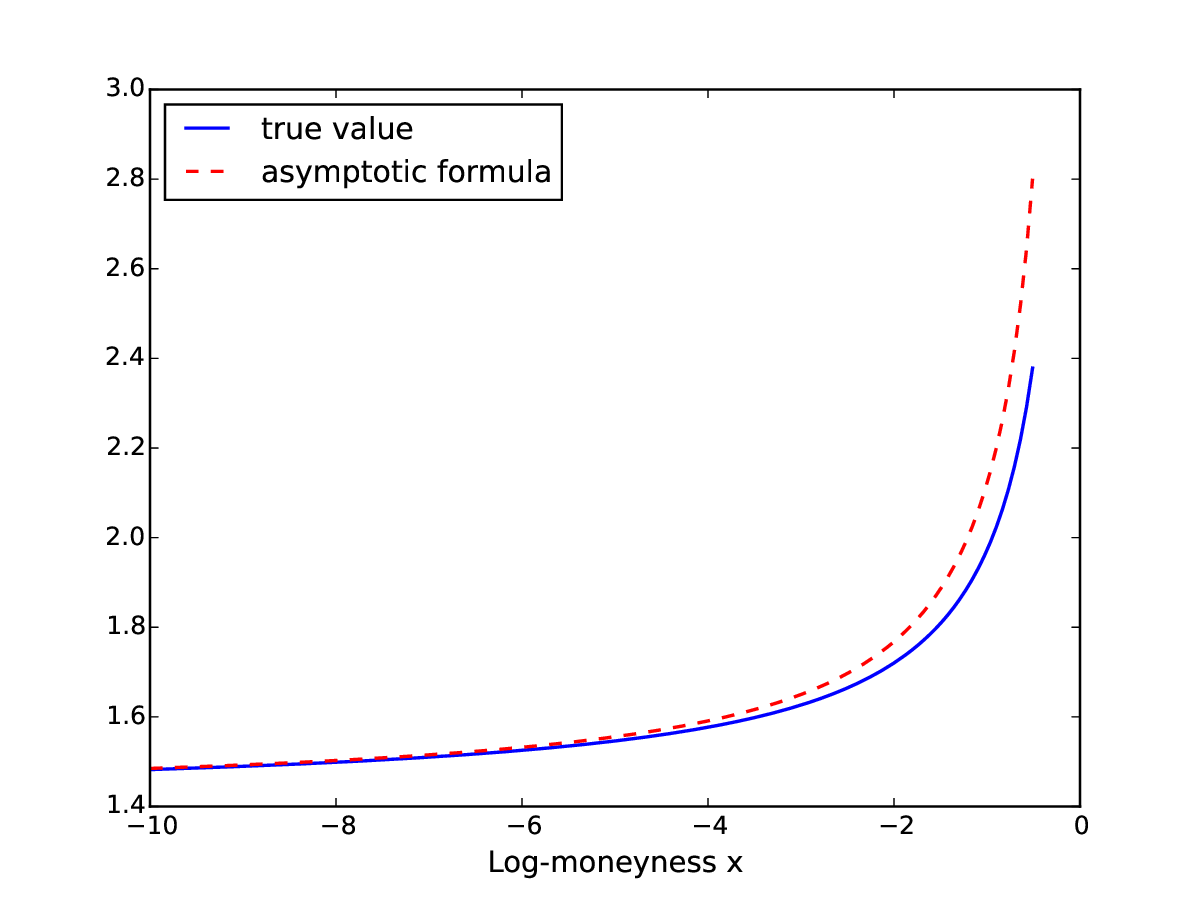}
           \graphicspath{ {./pics} }
        }
        \subfigure[Implied volatility smiles]{
        \label{fig:AffinePriceSmiles}
            \includegraphics[width=0.45\textwidth]{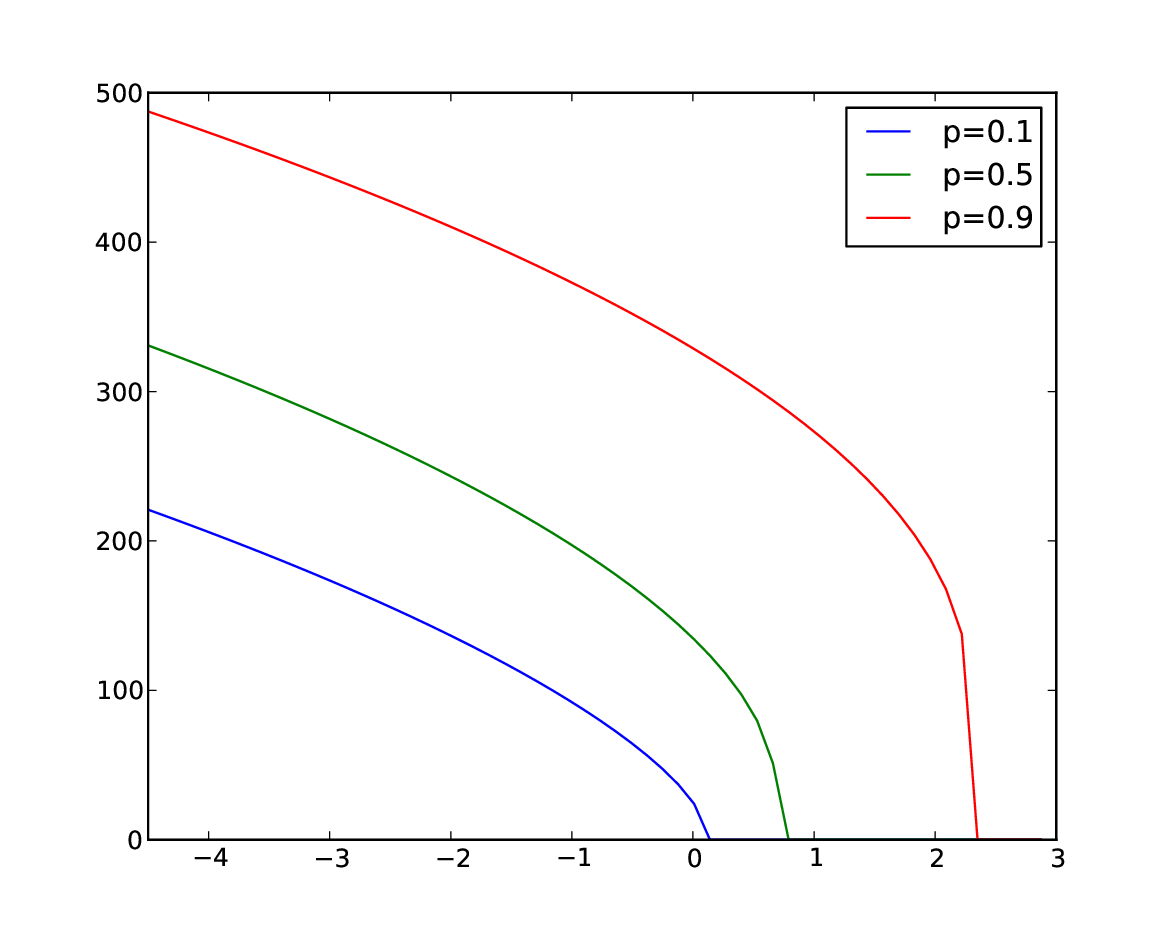}
            \graphicspath{ {./pics} }
        }\\ 
    \end{center}
		\caption{Implied volatilities generated by the affine Call price in Example~\ref{ex:affine}, $T=1$. 
(a)-(b)-(c): `normalised smile' $J(x) \equiv I(x) \sqrt{T/|x|}$ versus its approximation~$\tilde J(x)$ 
given in~\eqref{e:expansionPlot}.
Figure~(d): the corresponding implied volatilities (where the left derivative diverges at the upper bound of the support of the underlying).}
     \label{f:smilesAffine}
\end{figure}

\begin{figure}[!ht]
     \begin{center}
        \subfigure[Merton's model with mass $\pp$ at zero, $\pp=0.05$]{
        \label{fig:BlackScholesAtomSmiles}
            \includegraphics[width=0.45\textwidth]{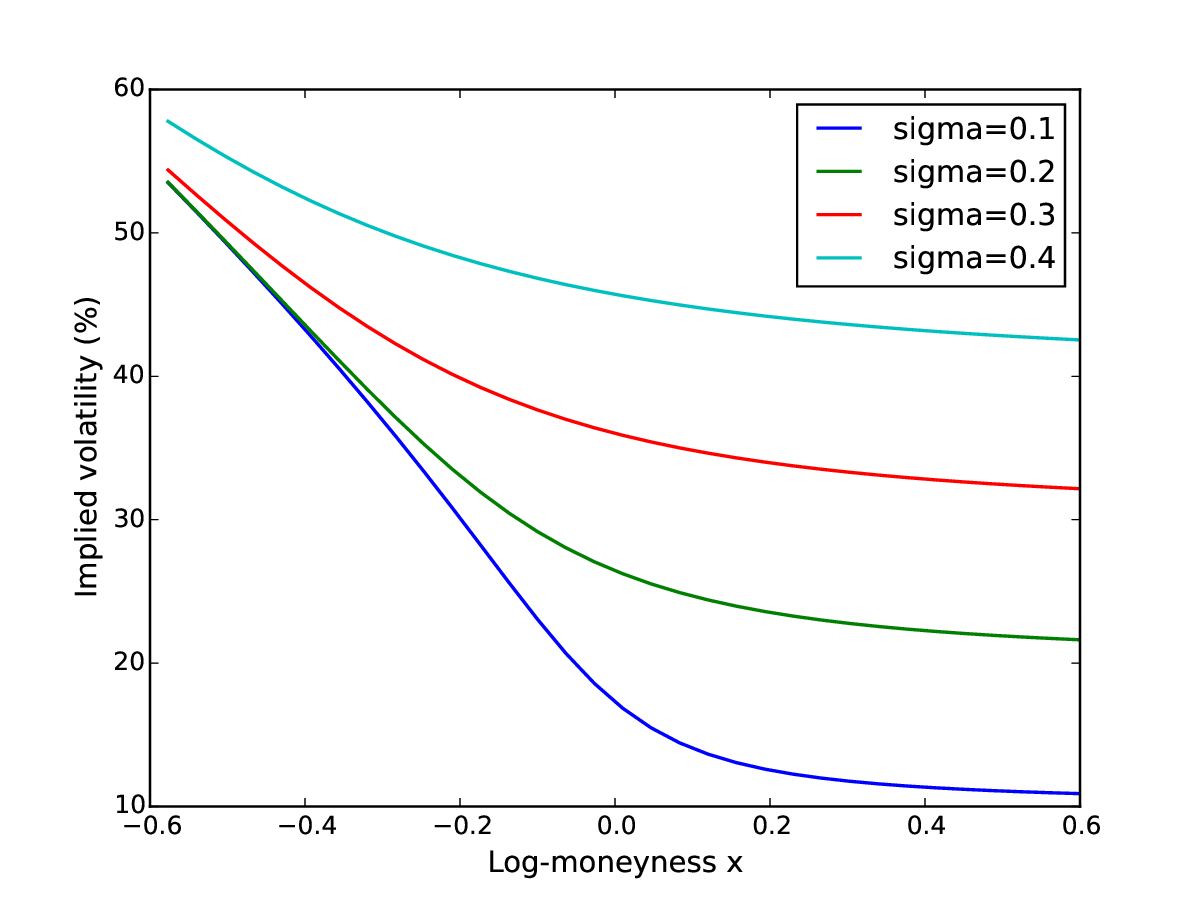}       
        }
        \subfigure[Merton's model with mass $\pp$ at zero, different $\pp$'s]{
        \label{third-BlackScholesAtomSmilesFixedSigma}
           \includegraphics[width=0.45\textwidth]{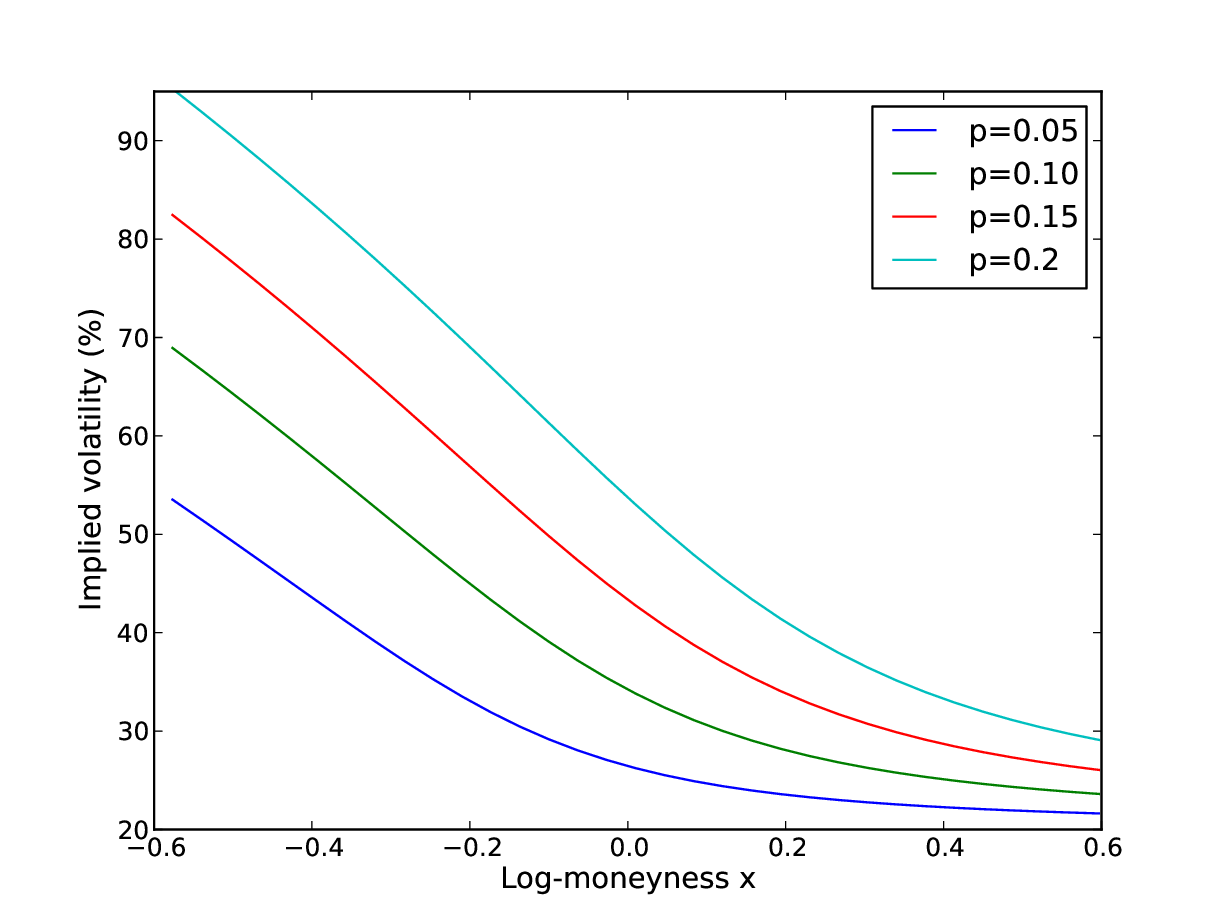}
        }\\ 
    \end{center}
		\caption{Implied volatility smiles in the Merton's model, or Black-Scholes distribution with mass $\pp$ at zero, with $S_0=1$, $T=1$. 
Figure (a): $\pp=0.05$, different values of $\sigma$. 
Figure (b): $\sigma=0.2$, different values of~$\pp$.} 
     \label{f:smilesBSAtom}
\end{figure}

\begin{figure}[t]
     \begin{center}
        \subfigure[Normalised smile, $\pp=0.1$]{
        \label{fig:BlackScholesAtomApprox01}
            \includegraphics[width=0.31\textwidth]{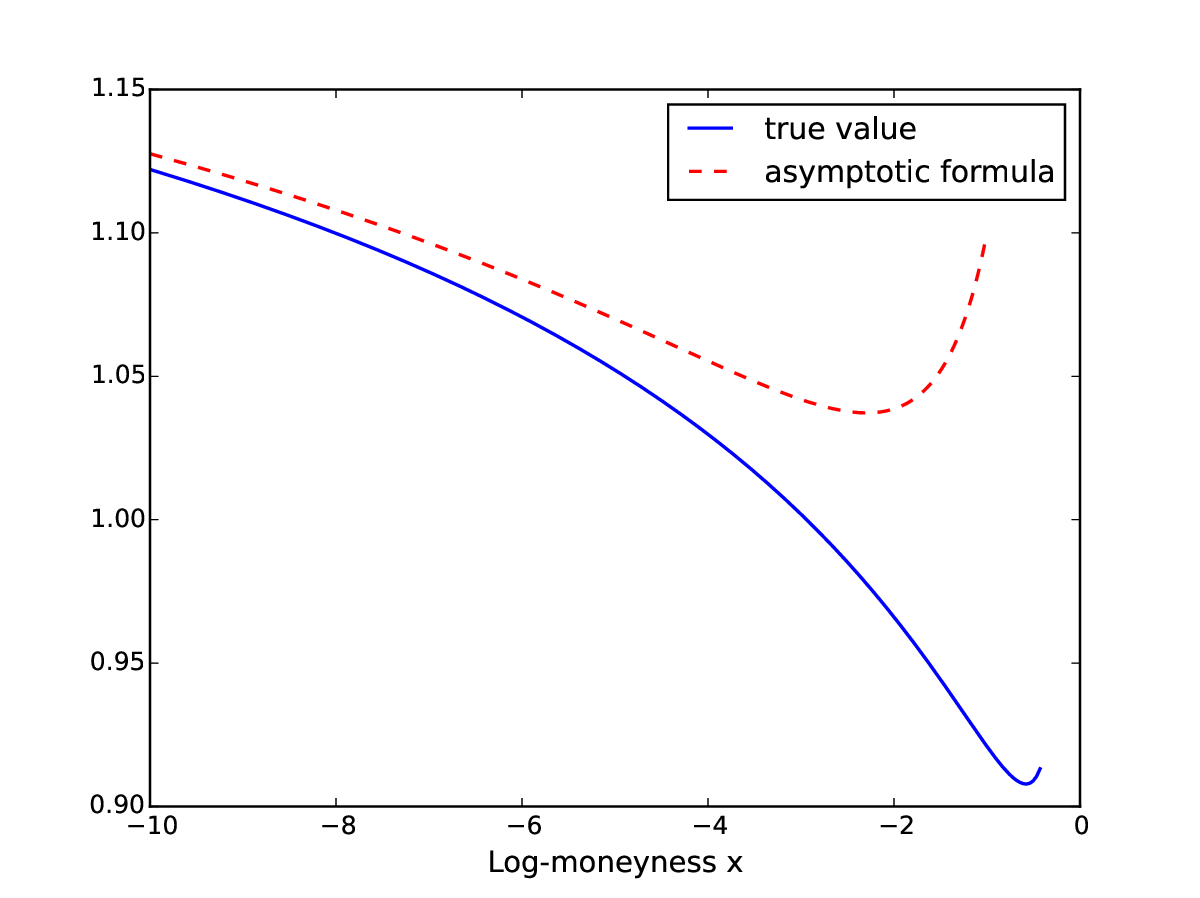}
        }
        \subfigure[Normalised smile, $\pp=0.5$]{
        \label{fig:BlackScholesAtomApprox05}
           \includegraphics[width=0.31\textwidth]{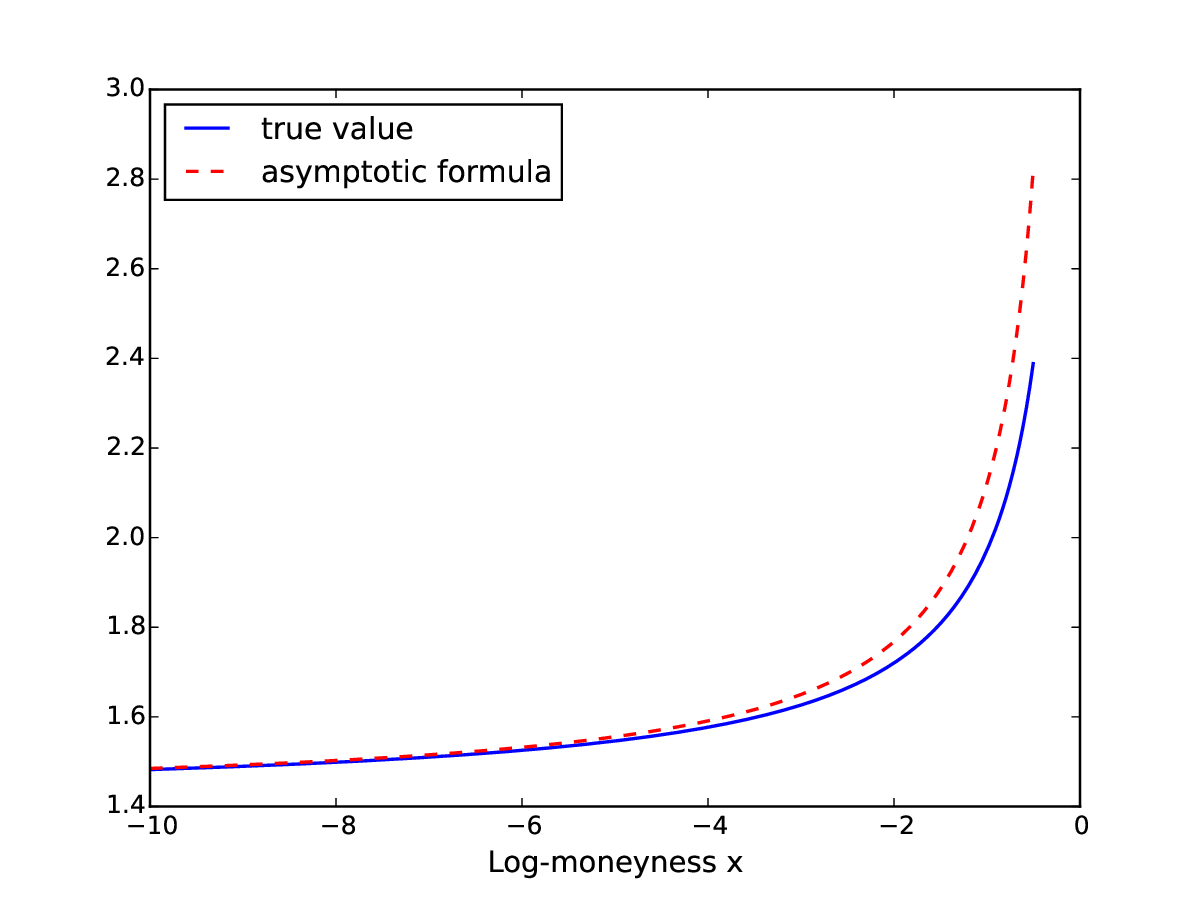}
        }
        \subfigure[Normalised smile, $\pp=0.9$]{
        \label{fig:BlackScholesAtomApprox09}
           \includegraphics[width=0.31\textwidth]{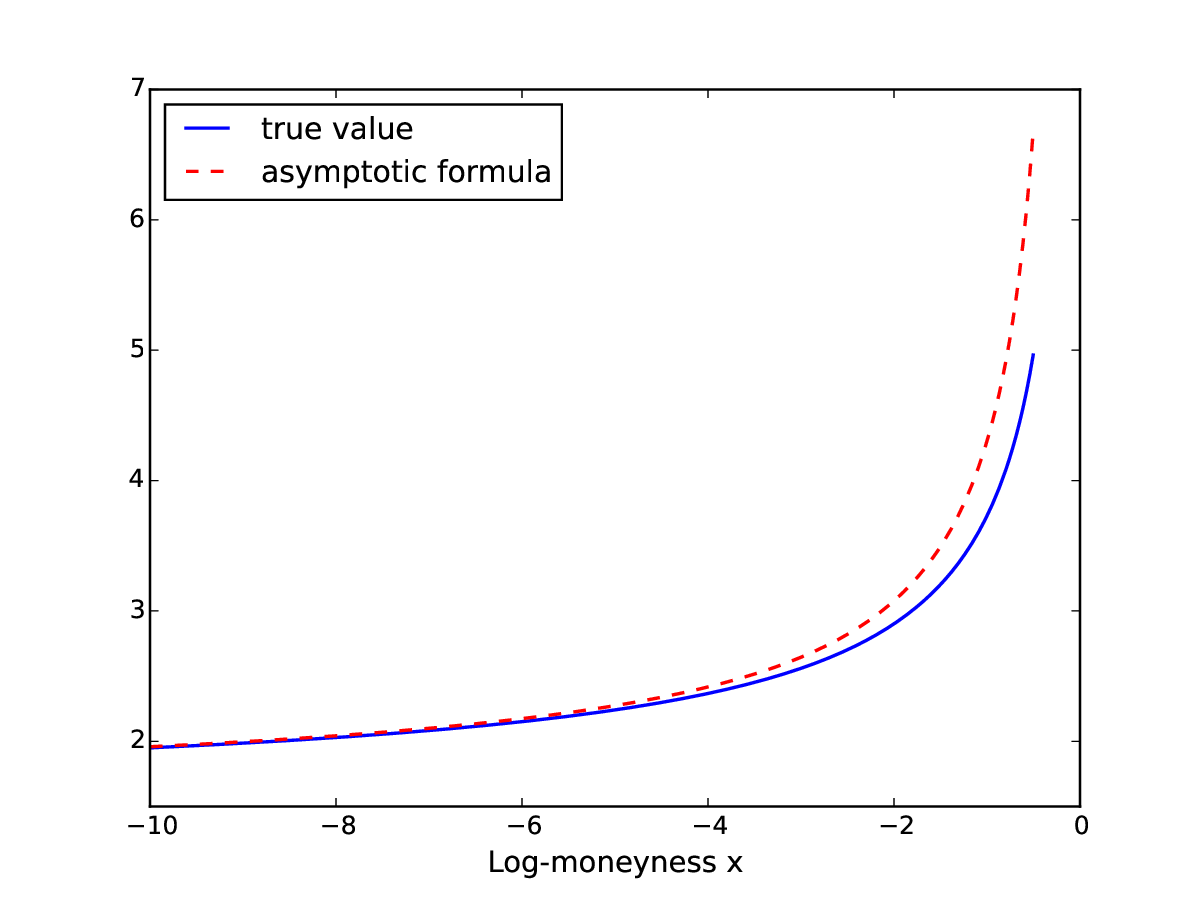}
        }\\ 
    \end{center}
\caption{Normalised implied volatility smiles $J(x) \equiv I(x) \sqrt{T/|x|}$ in the Merton's model with $\sigma=0.2$, $T=1$ and mass $\pp$ at zero. 
Comparison of the function $J$ with~$\tilde J$, see~\eqref{e:expansionPlot}.}
     \label{f:approxBSAtom}
\end{figure}

\subsection{Jump-to-default models} \label{s:j2d}
Let~$(\tilde{S}_t)_{t\geq 0}$ be a strictly positive process defined on~$(\Omega, \mathcal F, \Prob)$, and~$\tau$ a random time, independent of $\tilde{S}$.
Set
\be \label{e:j2d}
S_t = \tilde{S}_t \mathbf{1}_{\{t<\tau\}},
\ee
so that~$S$ jumps to zero at time $\tau$;
the law of~$S_T$ has the form~\eqref{e:distribAtom}, with $\pp=\Prob(\tau \le T)$, 
and~$\tilde \mu$ is the law of~$\tilde{S}_T$.

\subsubsection{Merton's model with jump-to-default} \label{s:merton0}

In the Merton model~\cite{Merton76optionpricing}, the process $\tilde{S}$ is a geometric Brownian motion with drift $\lambda > 0$: 
$\D \tilde{S}_t = \tilde{S}_t(\lambda \D t + \sigma \D W_t)$ with $\tilde{S}_0>0$, and~$\tau$ is exponentially distributed with parameter $\lambda$, so that $\pp=\Prob(\tau \le T)=1-\E^{-\lambda T}$.
Note that $\esp[S_T]= \esp[\tilde{S}_T   \ind_{\tau> T}] = \esp[\tilde{S}_T]   \Prob(\tau> T)   =S_0\E^{\lambda T}\E^{-\lambda T} = S_0$.
The continuous part of the distribution of~$S$ is $\tilde \mu(\D s) = f_{\BS}(s,S_0/(1-\pp),\sigma) \D s$,
where $f_{\BS}(\cdot,\overline{S},\sigma)$ is the density of a Black-Scholes stock price 
with mean~$\overline{S}$ and volatility $\sigma>0$.
The Put price written on~$S$ reads
$P(K) = \pp K + (1-\pp) P_{\BS}(K,S_0/(1-\pp),\sigma)$,
and the cumulative distribution
$$
F(K_x)  = \pp + (1-\pp) \Nn\left( -d_2\left(x+\log(1-\pp),\sigma\right)\right)
 =: \pp + (1-\pp) \Nn(d_{2,\pp}(x,\sigma)),
$$
where $d_{2,\pp}(x,\cdot) \equiv -d_{2}(x+\log(1-\pp), \cdot)$. 
Since $d_{2,\pp}(x,\sigma) \sim \frac{x}{\sigma\sqrt{T}}$ as $x \downarrow -\infty$ and $F(K_x)-F(0) = (1-\pp)\Nn(d_{2,\pp}(x,\sigma))$, 
the well-known bound $\mathcal{N}(d) \le \frac{n(d)}{|d|}$ for $d < 0$ yields
\[
F(K_x)-F(0) \le \frac{1-\pp}{|d_{2,\pp}(x,\sigma)|} n(d_{2,\pp}(x,\sigma))
 \le \frac{1-\pp}{\sqrt{2\pi}} \exp\left(-\frac{1}{2}d_{2,\pp}(x,\sigma)^2\right)
= \Oo\left(\exp\left(-\frac{x^2}{2 \sigma^2 T}\right) \right),
\]
and Condition~\eqref{e:archilCdf} is satisfied.
We illustrate the validity of~\eqref{e:explicitExp} in Figure~\ref{f:approxBSAtom}.
Let us briefly comment on Figures~\ref{f:smilesAffine}-\ref{f:approxBSAtom} related to the example~\eqref{e:affineCall} of an affine Call price
and to the Merton example above:
\begin{itemize}
\item[(i)] Interestingly, a log-normal distribution with a constant volatility parameter $\sigma$ and a mass $\pp$ at zero produces a very pronounced skew, even for relatively small values of~$\pp$
(Figure~\ref{f:smilesBSAtom}).
This is related to the impact of a mass at zero on the smile that we studied quantitatively 
in Theorem~\ref{t:impactOnSmile}(ii).
In analogy with displacement~\cite[Example 6.7]{CarrLee}, 
this is a way of generating a smile with only two parameters.

\item[(ii)] The different behaviours of the implied volatility foreseen by Corollary~\ref{c:secondOrder} (and by Theorem~\ref{t:IVArchil}) for $\pp \neq 1/2$ and for $\pp=1/2$ are confirmed in these examples 
(see also the CEV model in Section~\ref{s:cev}).
When $\pp=1/2$, the convergence of the normalised smile $I(x)\sqrt{T/|x|}$ to its limit~$\sqrt{2}$ 
has a considerably smaller bias than when~$\pp$ is close to one or to zero, 
for which the limiting value~$\sqrt{2}$ is still far in the left tail 
(Figures~\ref{fig:AffinePriceApproxP=0.1}-\ref{fig:AffinePriceApproxP=0.9} and~\ref{fig:BlackScholesAtomApprox01}-\ref{fig:BlackScholesAtomApprox09}).

\item[(iii)] The graphics~\ref{fig:AffinePriceApproxP=0.9},~\ref{fig:AffinePriceApproxP=0.5} and~\ref{fig:BlackScholesAtomApprox09},~\ref{fig:BlackScholesAtomApprox05} are almost identical. 
This provides evidence of the fact that the behaviour of the implied volatility for small strike is essentially determined by the mass of the atom at zero, while the remaining distribution on $(0,\infty)$ has little impact.
\end{itemize}

\subsubsection{Default probabilities from implied volatilities} \label{s:merton}
We consider here the same model as in Section~\ref{s:merton0}.
Ohsaki et al.~\cite{OOUY} study the possibility of measuring default probabilities from observed implied volatilities.
Considering a firm's asset following Merton's or CreditGrades~\cite{creditGrades} model, they estimate the survival probability at time $T$ based on the asymptotic formula 
$\lim_{x \downarrow -\infty} d_2(x) = -\qq$, computing~$d_2(x)$ from simulated smile data.
They give evidence of the difficulty of achieving a good estimate, due to the slow convergence of~$d_2$ to its limit.
For example, for a survival probability around $90\%$, the estimated value for the Merton model~\cite[Table 5]{OOUY} is affected by a relative error around $10\%$, even at extremely low strikes.

In Lemma~\ref{l:d2Estimate}, we account for the error term affecting this estimate, which is $\Oo(|x|^{-1/2})$ under Condition~\eqref{e:assumptionsCdf}(i).
Note however that Theorem~\ref{t:IVArchil} provides an alternative way of estimating default probabilities, 
which can be compared to the methodology in~\cite{OOUY}.
Neglecting the $\Oo(|x|^{-3/2})$ error term and inverting~\eqref{e:explicitExp} with respect to~$\qq$ 
yields the quadratic equation
$a(x) \qq^2 + b(x)\qq + c(x) = 0$,
with 
$a(x) \equiv 1/(2\sqrt{2|x|})$,
$b(x) = 1 + 1/(4|x|)$
 and $c(x) \equiv \sqrt{2|x|} -\sqrt T I(x) + 1/\sqrt{2|x|}$.
For $x$ small enough, the latter equation admits the two real roots
\be \label{e:solutionInversion}
\qq_{\pm}(x) = \frac{1}{\sqrt{2|x|}}\left\{-\frac{1}{2} - 2|x|
 \pm \sqrt{4|x|^{3/2}\left(I(x)\sqrt{2T}-\sqrt{|x|}\right)-2|x|+\frac{1}{4}}\right\}.
\ee
Using~\eqref{e:explicitExp}, it is not difficult to see that $\qq_{+}(x)$ converges to $\qq = \Nn^{-1}(\pp)$ while $\qq_{-}(x)$ diverges to infinity as $x \downarrow -\infty$, and hence $1-\Nn(\qq_+(x))$ is a convergent estimator of the survival probability $1-\pp$, independent of any parametric modelling choice.
\\
Table~\ref{table1} shows some numerical values for the Merton model.
The parameters are taken from~\cite{OOUY}:
\begin{equation}\label{param:Merton} 
\begin{aligned}
&S_0=100, \quad T=0.5, \quad \sigma=0.3,
\\
&\mbox{Set 1:} \quad \lambda=0.85 \qquad \mbox{Set 2:} \quad \lambda=0.15.
\end{aligned}
\end{equation}
For each parameter set, the first row shows the survival probability computed from the asymptotic formula
$\lim_{x \downarrow -\infty} d_2(x) = -\qq$, and provides the same values given in~\cite[Table 5]{OOUY}.
The second row shows the values of $1-\Nn(\qq_+(x))$.
The symbol `$-$' indicates that the estimator~$qq_+(x)$ is not defined 
(that is, the quadratic equation for~$\qq$ given above~\eqref{e:solutionInversion} does not have any solution).
The rightmost column contains the exact survival probability $1-\pp=\E^{-\lambda T}$.
The estimate based on the new formula~\eqref{e:explicitExp} proves to be more accurate: for example, for a moneyness ratio $K/S_0$ equal to $0.1$, the relative error is divided by three in the case of parameter Set 2 (roughly from $10\%$ to $3.5\%$), and divided by a factor $8$ (from $24\%$ to $3\%$) in the case of parameter Set 1.
Even if the fit is improved, the applicability of the model-free formula~\eqref{e:solutionInversion} for the estimation of default probabilities from market data can still be questioned. 
In this example, the relative error is below a few percents only for values of the strike/spot ratio outside the range usually observed in stock markets.

\begin{table}[ht]
	\centering
		\begin{tabular}{l c c c c c c c | r}
		\\ \hline \hline 
		Moneyness & 0.5 & 0.4 & 0.3 & 0.2 & 0.1 & 0.05 & 1e-10 & 0
		\\
		Log-moneyness $x$ & -0.69 & -0.91 & -1.20 & -1.61 & -2.30 & -3 & -23.02 & $-\infty$
		\\ [0.5ex] \hline 
		Survival Probability (\%): Set 1 
		\\ [0.5ex]  \hline
		\hspace{10mm} Ohsaki et al. & 42.13 & 43.76 & 45.43 & 47.23 & 49.44 & 51.01 & 59.80 & 65.37
		\\ [0.2ex]
		\hspace{10mm} $1-\Nn(\qq_+(x))$ & 71.43 & 70.26 & 69.23 & 68.29 & 67.37 & 66.86 & 65.48 & 65.37
		\\ [0.5ex] \hline \hline
		Survival Probability (\%): Set 2 
		\\ [0.5ex]  \hline
		\hspace{10mm} Ohsaki et al. & 75.59 & 77.77 & 79.72 & 81.59 & 83.59 & 84.86 & 90.35 & 92.77
		\\ [0.2ex]
		\hspace{10mm} $1-\Nn(\qq_+(x))$ & - & - & - & - & 96.15 & 95.04 & 92.90 & 92.77
		\\ [0.5ex] \hline \hline
		\end{tabular}
	\caption{Survival probabilities in Merton's jump-to-default model 
with the two parameters sets in~\eqref{param:Merton}.}
	\label{table1}
\end{table}

\begin{remark}
In view of~\eqref{e:putAsympt}, default probabilities could be estimated directly from Put prices by running a linear regression
$P(K) = \beta K+\varepsilon$ for small strikes, where the estimator for~$\beta$ would be an estimate
of the mass at the origin.
\end{remark}

\subsection{Diffusion processes absorbed at zero}

\subsubsection{The CEV process and comparison with Gulisashvili's formula~\cite{ArchilAtom}} \label{s:cev}

We consider here the CEV model, namely the unique strong solution to the stochastic differential equation
\begin{equation}\label{sdeCEV}
 \D S_t = \sigma S_t^{1+\beta}\D W_t,
\end{equation}
The process $(S_t)_{t\geq 0}$ is a true martingale~\cite[Chapter 6.4]{JYC2009} if and only if $\beta\leq 0$.
When $\beta=0$, the SDE~\eqref{sdeCEV} reduces to the Black-Scholes SDE, 
and the stock price remains strictly positive almost surely for all $t\geq 0$.
Following~\cite[Section 6.4]{JYC2009} we define a new process~$X$ 
by $X_t \equiv S_t^{-2\beta}/(\sigma^2 \beta^2)$ up to the first time~$S$ hits zero.
It\^o's formula yields
$\D X_t = \delta \D t + 2\sqrt{X_t}\D W_t$,
with $X_0 = S_0^{-2\beta}/(\sigma^2 \beta^2)>0$ and $\delta = 2+1/\beta$.
The process~$X$ is a Bessel process with~$\delta$ degrees of freedom 
and index $\nu \equiv \delta/2-1 = 1/(2\beta)$.
The Feller classification (see for example Karlin et al.~\cite[Chapter 15, Section 6]{Karlin} yields the following:
\begin{itemize}
\item if $\delta\leq 0$, i.e. $\beta\in [-1/2,0)$, the origin is an attainable and absorbing boundary.
For every $t>0$, the distribution $\mu_t$ of $X_t$ on $[0,\infty)$ has a positive mass at zero and admits a density on the positive real line:
\[
\mu_t(\D y) = \Prob(X_t=0) \delta_{0}(\D y) + f_{X_t}(X_0,y) \D y,
\]
with
$$
f_{X_t}(X_0,y) = \frac{1}{2t}\left(\frac{y}{X_0}\right)^{\nu/2}\exp\left(-\frac{X_0+y}{2t}\right)
I_{-\nu}\left(\frac{\sqrt{X_0 y}}{t}\right),
\qquad\text{for all } y>0,
$$
where $I_{-\nu}$ is the modified Bessel function of the first kind.
Note that
$\int_{0}^{\infty}f_{X_t}(X_0,y) \D y = \Gamma\left(-\nu, \frac{X_0}{2t}\right)<1$,
where $\Gamma$ is the normalised lower incomplete Gamma function
$\Gamma(v,z) \equiv \frac{1}{\Gamma(v)}\int_{0}^{z}u^{v-1}\E^{-u}\D u$,
therefore
$\Prob(X_t=0) = 1-\Gamma\left(-\nu,X_0/(2t)\right) > 0$;
\item if $\delta \in (0,2)$ ($\beta<-1/2$), the origin is attainable.
If $\delta >2$ ($\beta > 0$), the origin is not attainable. In both cases, $\Prob(X_t=0)=0$ for all $t$.
\end{itemize}

We can recast these results in terms of the original CEV process~$S$,
which hits zero if and only if the process~$X$ does.
In the case $\beta\in [-1/2,0)$, the density of $S_T$ on the positive real line is given by
$$
f_{S_T}(s)= 
-\frac{S_0^{1/2}s^{-2\beta-3/2}}{\sigma^2 \beta T}
\exp\left(-\frac{S_0^{-2\beta}+s^{-2\beta}}{2\sigma^2\beta^2 T}\right)
I_{-\nu}\left(\frac{S_0^{-\beta}s^{-\beta}}{\sigma^2 \beta^2 T}\right),
$$
for any $s>0$, and we further have
$\pp = \PP(S_T =0) = 1 - \Gamma\left(-\nu,(2\sigma^2\beta^2 T s_0^{2\beta})^{-1}\right)$.
Using the asymptotic form~\cite[Section 9.6.7]{AbrSteg} for the modified Bessel function 
$I_{\alpha}(z) \sim \Gamma(\alpha+1)^{-1}(z/2)^{\alpha}$ (as $z\downarrow 0$) for positive  $\alpha$, together with $-\nu=1/(2|\beta|)$, one obtains $f_{S_T}(s) \sim const \times s^{2|\beta|-1}$ as $s\downarrow 0$.
Therefore the density of the stock price explodes at the origin when $\beta \in (-1/2,0)$, and tends to a constant when $\beta=-1/2$, in contrast to the previous examples (where the density vanishes at the origin).
As pointed out in the discussion right after~\eqref{e:measure}, Condition~\eqref{e:archilCdf} on the cumulative distribution is satisfied here since $2|\beta|-1>-1$.
This CEV model can further be enhanced with an additional non-predictable independent jump-to-default, as done in~\cite{CampiSbuelz}.
This would result in augmenting the mass at zero and reducing the one on $(0,\infty)$, without affecting the shape of the density.

In order to test our results numerically, we need first to compute the price of European Put options,
which, for a maturity $T\geq 0$ and a strike $K\geq 0$, is given by
\be \label{e:cevPut}
P(K) = \mathbb{E}[(K-S_T)_+]
 = \pp K + \int_{(0,+\infty)}(K-s)^+ f_{S_T}(s)\D s.
\ee
We then provide a numerical comparison of our mass-at-zero approximation of the implied volatility smile to the true one.
More precisely, we compare the formulas~\eqref{e:explicitExp} and~\eqref{e:archilAtom} to
the true implied volatility smile computed from the direct integration formula~\eqref{e:cevPut}.
The `2-term approximation' in Figures~\ref{f:CEVPlotsSmallMass} and~\ref{f:SmilesMassPlot} 
refers to~\eqref{e:explicitExp} when considering the terms
up to order $|x|^{-1/2}$, 
and the `3-term approximation' corresponds to all the terms in the formula up to order $|x|^{-1}$.
We shall consider several cases, depending on the magnitude of the mass at zero.
In Figure~\ref{f:CEVPlotsSmallMass}, the mass at zero is small ($\pp\approx 1.47\%$), 
whereas a situation with large mass at zero ($\pp\approx 71.89\%$) can be observed in Figure~\ref{f:CEVPlotsBigMass}.
As one can see, our explicit asymptotic formula~\eqref{e:explicitExp} approaches quickly the non-explicit formula~\eqref{e:archilAtom}, providing a good approximation of the implied volatility smile for small log-moneyness.
Furthermore, in Figure~\ref{f:SmilesMassPlot}, we compare~\eqref{e:explicitExp} 
with Gulisashvili's~\eqref{e:archilAtom} as a function of the maturity of the option 
(for two different values of the log-moneyness~$x$).

\begin{figure}[!ht]
     \begin{center}
        \subfigure[Implied volatility smiles in the CEV model]{
            \includegraphics[height=4.5cm,width=0.48\textwidth]{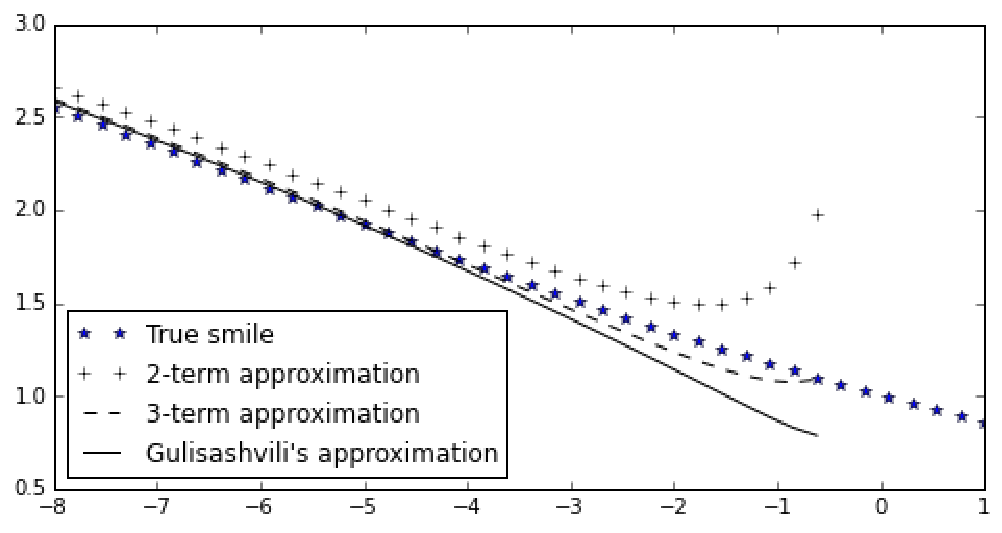}
        }
        \subfigure[Errors]{
           \includegraphics[height=4.5cm,width=0.48\textwidth]{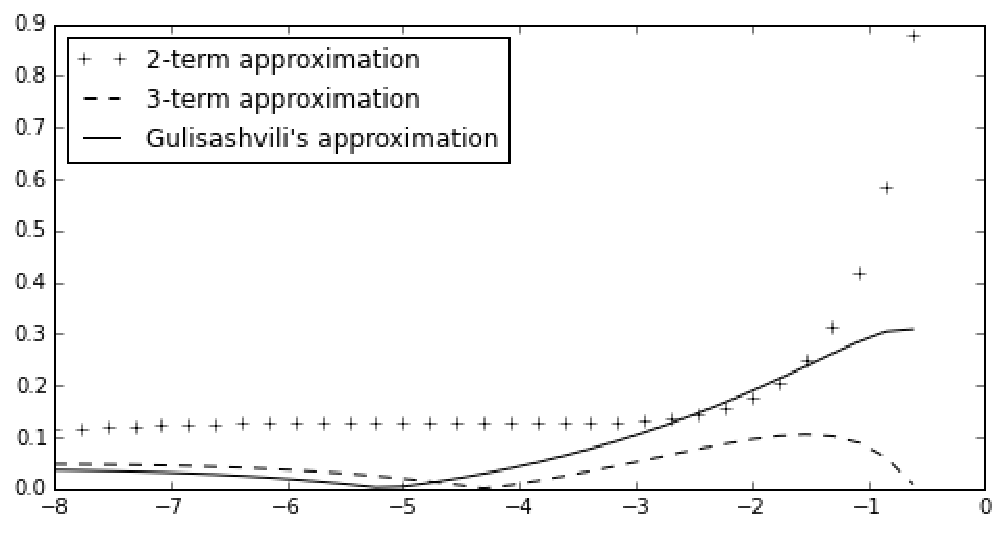}
        }
    \end{center}
\caption{Comparison of the implied volatility smiles in the CEV model.
Here, 
$s_0 = 0.1$, $T = 1$, $\beta = -0.3$, $\sigma = 0.5$,
giving a mass at zero $\pp\approx 1.47\%$.
Gulisashvili's approximation in solid line is~\eqref{e:archilAtom} and the dashed `3-term approximation' is the new formula~\eqref{e:explicitExp}.
}
     \label{f:CEVPlotsSmallMass}
\end{figure}

\begin{figure}[!ht]
     \begin{center}
        \subfigure[Implied volatility smiles in the CEV model]{
            \includegraphics[height=4.5cm,width=0.48\textwidth]{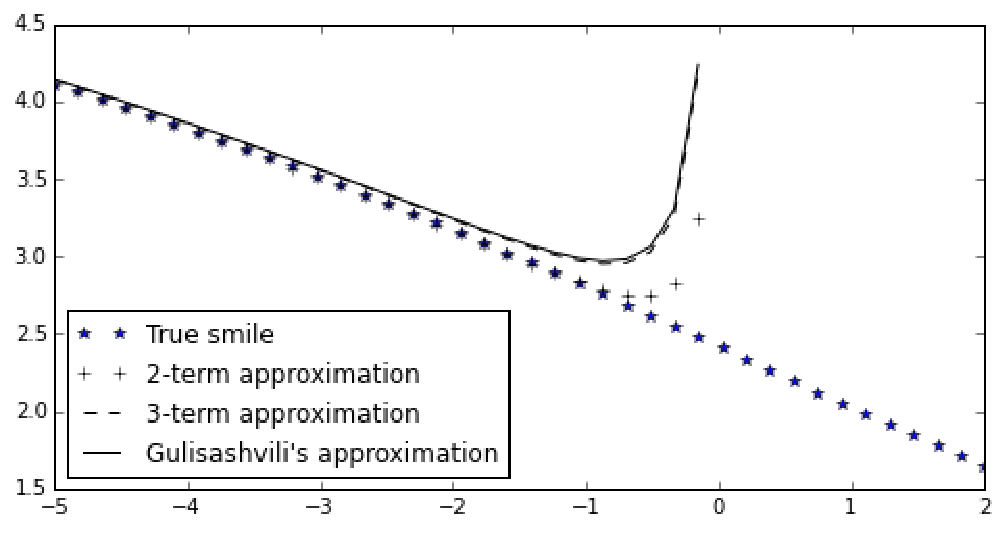}
        }
        \subfigure[Errors]{
           \includegraphics[height=4.5cm,width=0.48\textwidth]{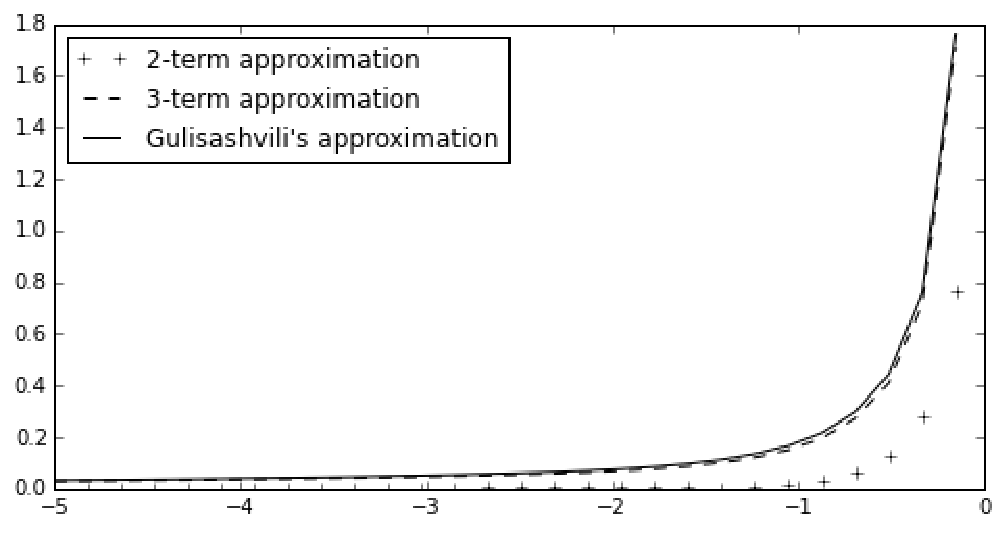}
        }
    \end{center}
\caption{Comparison of the implied volatility smiles in the CEV model.
Here, 
$s_0 = 0.1$, $T = 1$, $\beta = -0.4$, $\sigma = 1$,
giving a mass at zero $\pp\approx 71.89\%$.
Gulisashvili's approximation in solid line is~\eqref{e:archilAtom} and the dashed `3-term approximation' is the new formula~\eqref{e:explicitExp}.
}
     \label{f:CEVPlotsBigMass}
\end{figure}

\begin{figure}[!ht]
\centering
            \includegraphics[scale = 0.6]{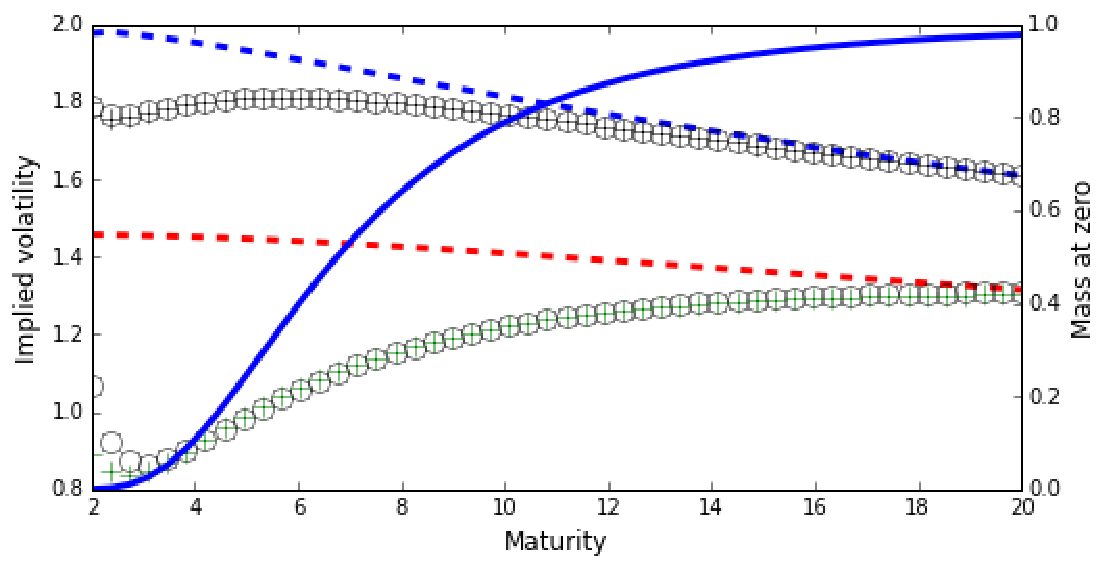}
\caption{Comparison of the implied volatility smiles in the CEV model.
Here, we take 
$s_0 = 0.1$, $\beta = -0.1$, $\sigma = 1$, and we let the maturity vary from two to twenty years.
The solid line (with values on the right vertical axis) represents the mass at zero.
The dashed lines, circles and crosses denote respectively the true smile, Gulisashvili's formula~\eqref{e:archilAtom}
 and our new approximation~\eqref{e:explicitExp} (all with values on the left vertical axis). 
The top three graphs corresponds to taking log-moneyness $x=-5$, and the bottom three $x=-2$.
}
     \label{f:SmilesMassPlot}
\end{figure}

\subsubsection{Absorbed Ornstein-Uhlenbeck process}

Another example of continuous asset price dynamics that accumulates mass at zero and allows for explicit formulae can be built from Ornstein-Uhlenbeck (OU) processes, namely
the unique strong solutions to the SDE
$\D \tilde{S}_t = -k \tilde{S}_t  \D t +  \sigma \D W_t$, with $\tilde{S}_0 = s_0 >0$ and $k,\sigma>0$.
Then $\tilde{S}_t=s_0 \exp(-kt) + \sigma \int_0^t  \exp(-k(t-u))  \D W_u$
 is a Gaussian process with mean
$ \mathbb{E}(\tilde{S}_t)= s_0 \exp(-kt)$
and covariance function
${\rm Cov}(\tilde{S}_t,\tilde{S}_s)=\frac{\sigma^2}{k}\exp(-kt) \sinh(ks)$.
The origin is attainable, and we define~$S$ as~$\tilde{S}$ stopped at the first time it hits zero:
let $\tau_0 \equiv \inf \{t \geq 0: \tilde{S}_t=0\}$, then
$S_t \equiv \tilde{S}_t \ind_{\{t < \tau_0\}}$.
For every $t>0$, the law of~$S_t$ has the form $\mu_t(\D y) = \Prob(S_t=0) \delta_{0}(\D y) + f_t(s_0,y) \D y$; from Borodin and Salminen~\cite{handbook}, we have
\begin{equation} \label{e:OUatom}
\pp = \Prob(S_t =0)  = \Prob (\tau_0 \leq  t) = 
2\Nn
\left(-\frac{s_0}{\sigma \sqrt{\E^{2kt}-1}}\right),
\end{equation}
and
\be \label{e:OUdensity}
f_{t}(s_0,y) = \Prob\Bigl(\tilde{S_t} \in \D y, \min_{0 \leq s \leq t }\tilde{S_s} >0 \Bigr) =
\frac{2\sinh\left(2y s_0 \E^{-kt}\right)}{\sqrt{2 \pi \sigma^2(1-\E^{-2kt})}}
\exp\left(-\frac{y^2 + s_0^2 \E^{-2kt}}{2 \sigma^2(1-\E^{-2kt})}\right),
\qquad\text{for all } y>0.
\ee
Note that $f_{t}(s_0,y) \sim c_t \: y$ as $y\downarrow 0$, so that $F(K)-F(0) \sim c_t K^2$ as $K\downarrow 0$, and Condition~\eqref{e:archilCdf} is satisfied.
Using~\eqref{e:OUatom} and~\eqref{e:OUdensity}, the numerical evaluation of European options is straightforward from numerical integration of~$f_t$ or from Monte-Carlo simulation of OU paths; for small log-moneyness, the shape of the implied volatility smile is again described by Theorem~\ref{t:IVArchil} (and will eventually, in the limit as~$x$ tends to ~$-\infty$, be similar to the smile of the CEV and the other jump-to-default models above).

\subsection{Some comments on smile parameterisations}

It is worth noticing that most of the recent literature on implied volatility parameterisations does not seem to take into account the possibility of having a mass at the origin.
We discuss two arbitrage-free examples: the SSVI model \cite{SSVI}, 
and the parameterisation proposed in Guo et al.~\cite{GJMN2012}.
First of all, note the following for the total implied variance $w(x) \equiv T I^2(x)$:
if $\qq = N^{-1}(\pp) \neq 0$, then~\eqref{e:secondOrderAtom} implies
\be \label{e:secondOrder_totalVariance}
 w(x) - 2|x| = \left( \sqrt{T} I(x) - \sqrt{2|x|} \right) \left(\sqrt{T} I(x) + \sqrt{2|x|} \right)
\to \mathrm{sgn}(\qq) \infty =
\left\{
\begin{array}{ll}
-\infty & \mbox{if } 0 \mino \pp \mino 1/2, \\
+\infty & \mbox{if } \pp \maj 1/2.
\end{array}
\right.
\ee

\subsubsection{SSVI}
Gatheral and Jacquier~\cite{SSVI} suggest to model the total implied variance with the following family of functions:
\begin{equation} \label{SSVI}
w_{\mathrm{SSVI}}(x) = \frac{\theta}2 \left(1 + \rho \varphi x + \sqrt{(\varphi x + \rho)^2 + 1 -\rho^2} \right),
\end{equation}
where $\theta > 0$, $\varphi \maj 0$ and $\rho \in (-1,1)$ are three parameters.
From~\cite[Theorem 4.2]{SSVI}, the parameterisation~\eqref{SSVI} is free of arbitrage 
(for a given maturity~$T$) if both conditions 
$\theta \varphi (1+|\rho|) < 4$ and $\theta \varphi^2 (1+|\rho|) \le 4$ are satisfied; 
moreover, the condition 
$\theta \varphi (1+|\rho|) \le 4$ is shown to be necessary~\cite[Lemma 4.2]{SSVI}.
It is straightforward to see that $w_{\mathrm{SSVI}}(x)/|x| \to \frac{\theta}2 \varphi (1 \pm \rho)$ as $x \to \pm \infty$.
In light of the necessary condition for no arbitrage $\theta \varphi (1+|\rho|) \le 4$, in order to have the maximal slope $\lim_{x \downarrow -\infty} w_{\mathrm{SSVI}}(x)/|x| = 2$ for the left wing, we need to impose
\[
\rho \le 0
\quad \mbox{ and } \quad
\theta \varphi (1+|\rho|) = 4.
\]
The second condition above indicates that we are on the boundary of the admissible parameter set.
The following argument is taken from \cite[Section 7.2]{DM-Martini}: assuming $\theta \varphi (1+|\rho|) = \theta \varphi (1 - \rho) = 4$, it is not difficult to see that the following expansion holds:
\be \label{e:SSVI_secondOrder}
w_{\mathrm{SSVI}}(x) - 2|x| =
\frac{\theta}{2}(1 - \rho) + \mathcal{O}\left(\frac{1}{|x|}\right) \to \frac{\theta}{2}(1 - \rho)
\qquad \mbox{as } x \downarrow -\infty.
\ee
The limit above contradicts both cases in \eqref{e:secondOrder_totalVariance}.
As a conclusion, a positive mass $\pp \neq 1/2$ \emph{cannot} be embedded into the SSVI parameterisation while keeping the no-arbitrage conditions.

\subsubsection{The parameterisation by Guo et al.~\cite{GJMN2012}}
The proposed parameterisation is 
$w(x) = \theta \, \Psi(x\xi(\theta))$, 
where 
\[
\xi(u) \equiv \alpha\frac{1-\E^{-u}}{u}
\qquad\text{and}\qquad
\Psi(z) \equiv |z| +\frac{1}{2}\left(1+\sqrt{1+|z|}\right).
\]
with $\alpha, \theta>0$.
Since the expansion $w(x) = \alpha|x|(1-\E^{-\theta}) + \frac{1}{2}\theta + \mathcal{O}(|x|^{-1/2})$ holds as $x$ tends to $-\infty$, the asymptotic slope $\lim_{x \downarrow -\infty} \frac{w(x)}{|x|}$ is equal to~$2$ if and only if
$\alpha(1-\E^{-\theta})=2$.
But this entails
$$
\lim_{x \downarrow -\infty}\left(w(x) - 2|x|\right)  = \frac{\theta}{2},
$$
again contradicting \eqref{e:secondOrder_totalVariance}, and
therefore ruling out the possibility of a mass at zero.

\appendix

\section{Appendix}

\subsection{Put-Call duality and smile symmetry}\label{s:symmetry} 

Fix some maturity $T>0$. The function 
\be \label{e:Gfunction}
G(K) \equiv \frac{K}{S_0} P\left(\frac{S_0^2}{K}\right), \qquad K > 0,
\ee
allows to define a Black-Scholes implied volatility function $I_{G}$, when $G$ is taken as a Call price with maturity~$T$.
The identity
\be \label{e:GimpliedVol} 
I_{C}(K) = I_{G}\left(\frac{S_0^2}K\right)
\ee
is proven and used in~\cite{GulForm,GulIJTAF} to transfer the asymptotic results initially formulated for the right part of the implied volatility smile ($K \uparrow \infty$) to the left part ($K \downarrow 0$).

\begin{proposition}  \label{e:GisNotACall}
When $\pp>0$, the function $K \mapsto G(K)$ defined in~\eqref{e:Gfunction} is not a Call price function.
\end{proposition}

\begin{proof}
Assume~$G(\cdot)$ is a Call price function with maturity~$T$, 
then $G(K)=\esp(X-K)^+$ for some integrable random variable~$X$.
Equation~\eqref{e:putAsympt} implies
\[
\lim_{K \uparrow \infty} G(K)
=
\lim_{K \uparrow \infty} \frac{K}{S_0} P\left(\frac{S_0^2}{K}\right)
=
\lim_{K' \downarrow 0} \frac{S_0}{K'} P(K')
= \pp S_0 > 0,
\]
which contradicts $\lim_{K \uparrow \infty} G(K)=\lim_{K \uparrow \infty} \esp(X-K)^+ = 0$ by dominated convergence.
\end{proof}
\medskip

The situation where~$G$ is a genuine Call price function, and moreover $G \equiv C$, is related to a symmetry of the underlying law.
Denote by $\Q$ the probability measure defined by the Radon-Nikodym density $\D\Q/ \D\Prob = S_T/S_0$.
The distribution of~$S_T$ is said to be geometrically symmetric if the distribution of~$S_0/S_T$ 
under~$\Q$ is the same as the distribution of~$S_T/S_0$ under~$\Prob$ 
(see Carr and Lee~\cite{CarrLee}).
Examples include the log-normal distribution and uncorrelated stochastic volatility models (with zero risk-free rate).
It is easy to see~\cite[Theorem~2.2 and Corollary~2.5]{CarrLee} that geometric symmetry implies 
(and indeed is equivalent to) the Put-Call price symmetry
\be \label{e:pcSymm}
C(K) = G(K)
\ee
with $G$ defined in~\eqref{e:Gfunction}.
Note that~\eqref{e:pcSymm} can be also written in the more `symmetric' fashion $P(K,S_0)=C(S_0,K)$, making the spot price appear explicitly.
Equation~\eqref{e:GimpliedVol} shows that Put-Call symmetry is in turn equivalent to the symmetry of the implied volatility smile with respect to the log-moneyness
\be \label{e:symm}
I(x) = I(-x), \qquad \text{for all } x \in \R.
\ee
The equivalence of~\eqref{e:symm} and~\eqref{e:pcSymm} gives the following corollary to Proposition~\ref{e:GisNotACall}:

\begin{corollary} \label{c:noSym}
If $\pp>0$, the implied volatility at expiry $T$ cannot be symmetric in the sense of~\eqref{e:symm}.
\end{corollary}

\begin{remark}
Note that $\Q(S_T>0) = \esp_{\Prob}\left[(S_T/S_0) \ind_{S_T>0} \right] = 1$, therefore
$S_0/S_T$ is $\Q$-almost surely well-defined also when the $\Prob$-distribution of $S_T$ has an atom at zero.
However, since $\Q\left(S_0/S_T > 0\right)=1$, the $\Q$-distribution of $S_0/S_T$ cannot coincide with the $\Prob$-distribution of $S_T/S_0$ in this case.
This is another way of showing that geometric symmetry, hence symmetry of the smile, 
does not hold when the $\Prob$-distribution has an atom at zero.
\end{remark}

\begin{remark}
In~\cite[Theorem 4.1]{Lee}, assuming $\pp = 0$, Lee proves the identity $I^{\Prob}(x)=I^{\Q}(-x)$, 
where~$I^{\Q}$ denotes the implied volatility of options written on $S_0/S_T$ 
and priced under the measure~$\Q$.
Although both functions~$I^{\Prob}$ and~$I^{\Q}$ are well defined for any stock price distribution that is non-negative under~$\Prob$, the same argument used in the proof of Proposition~\ref{e:GisNotACall} shows that the identity $I^{\Prob}(x)=I^{\Q}(-x)$ does not hold when $\pp>0$.
\end{remark}


\bibliographystyle{siam}
\bibliography{References}

\end{document}